\documentclass[12pt, draftclsnofoot, onecolumn]{IEEEtran}

\usepackage{subfigure}
\usepackage{times}
\usepackage{latexsym}
\usepackage{morefloats}
\usepackage{amssymb}
\usepackage{amsmath}

\usepackage{enumitem}
\usepackage{mathrsfs}
\usepackage{amsgen,amsfonts,amsbsy,amsthm}
\usepackage{graphicx}
\usepackage{wasysym}
\usepackage{pxfonts}
\usepackage{stackrel}
\usepackage{array}
\usepackage{hyperref}
\newcommand{\remark}[1]{{\emph{Remark:}} #1}

\newcommand{\BEQA}{\begin{eqnarray}}
\newcommand{\EEQA}{\end{eqnarray}}

\newcommand{\gap}{\vspace{2mm}}

\newcommand{\x}{\mathbf{x}}

\newcommand{\hLambda}{\hat{\Lambda}}
\newcommand{\hhLambda}{\hat{\hat{\Lambda}}}
\newcommand{\tLambda}{\tilde{\Lambda}}
\newcommand{\pdel}{p_{\mathrm{del}}}
\newcommand{\dmax}{d_{\max}}
\newcommand{\deltabar}{\overline{\delta}}
\newcommand{\dbar}{\overline{d}}
\newtheorem{theorem}{Theorem}
\newtheorem{proposition}{Proposition}
\newtheorem{lemma}{Lemma}

\newtheorem{insight}{Insight}

\begin{document}

\title{An Approximate Inner Bound\\ to the QoS Aware Throughput Region\\ of a Tree Network under\\ IEEE 802.15.4 CSMA/CA and\\ Application to Wireless Sensor\\ Network Design}
\author{Abhijit Bhattacharya and Anurag Kumar\\Dept.\ of Electrical Communication Engineering\\Indian Institute of Science, Bangalore, 560012, India\\ email: \{abhijit, anurag\}@ece.iisc.ernet.in}


\maketitle

\begin{abstract}
We consider a tree network spanning a set of source nodes that generate measurement packets, a set of additional relay nodes that only forward packets from the sources, and a data sink. We assume that the paths from the sources to the sink have bounded hop count. We assume that the nodes use the IEEE~802.15.4 CSMA/CA for medium access control, and that there are no hidden terminals. In this setting, starting with a set of simple fixed point equations, we derive sufficient conditions for the tree network to approximately satisfy certain given QoS targets such as end-to-end delivery probability and delay under a given rate of generation of measurement packets at the sources (arrival rates vector). The structures of our sufficient conditions provide insight on the dependence of the network performance on the arrival rate vector, and the topological properties of the network. Furthermore, for the special case of equal arrival rates, default backoff parameters, and for a range of values of target QoS, we show that among all path-length-bounded trees (spanning a given set of sources and BS) that meet the sufficient conditions, a shortest path tree achieves the maximum throughput. 
\end{abstract}

\section{Introduction}
\label{sec:intro}
Our work in this paper is motivated by the following broad problem of designing multi-hop ad hoc wireless networks that utilise IEEE~802.15.4 CSMA/CA as the medium access control. Given a network graph over a set of sensor nodes (also called sources), a set of potential relay locations, and a data sink (also called base station (BS)), where each link meets a certain target quality requirement, the problem is to extract from this graph, a hop length bounded\footnote{For a discussion of why such a hop count bound is needed, see \cite{bhattacharya-kumar14comnet,bhattacharya-kumar14tr-qos-aware-nw-design-csma}} tree topology connecting the sensors to the BS, such that the resulting tree provides certain quality of service (QoS), typically expressed in terms of a bound on the packet delivery probability and/or on the mean packet delay, \emph{while also achieving a large throughput region}.

\begin{figure}[t]
\footnotesize
\begin{center}
\includegraphics[scale=0.35]{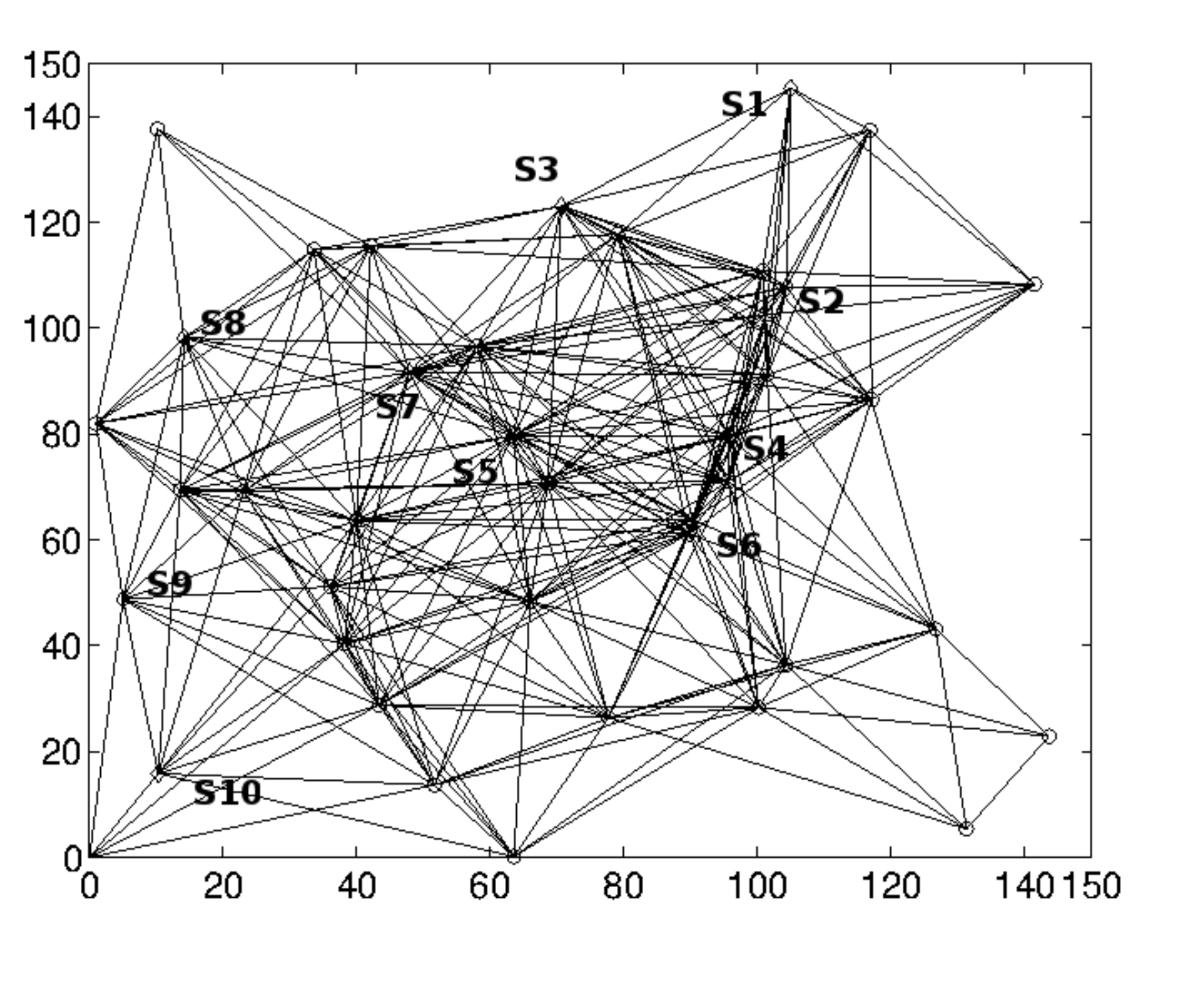}
\caption{A network graph over 10 sources and 30 potential relay locations; the base station (BS) is at (0,0); each edge is assumed to have a packet error rate of no more than 2\%. There are no hidden nodes.}
\label{fig:illustrative-topological-example}
\vspace{-5mm}
\end{center}
\normalsize
\end{figure}
As an example, consider the network graph shown in Figure~\ref{fig:illustrative-topological-example} over 10 sensors, 30 potential relay locations, and a BS at $(0,0)$; the links in the network have a worst case packet error rate of 2\%. Suppose that the nodes use IEEE~802.15.4 CSMA/CA, and that all the nodes are within carrier sense range of one another. The problem is to obtain from this graph, a tree connecting the sensors to the sink, such that the hop count from each sensor to the sink is no more than, say 5, the packet discard probability on each link is no more than 2.08\%, and the mean delay on each link is no more than 20 msec (these single hop QoS requirements translate to an end-to-end delivery probability of 90\%, and an end-to-end mean delay of 100 msec). In addition, among trees that meet these requirements, the resulting tree should achieve a large throughput region. The following are possible approaches for addressing such a problem. 

\gap
\noindent
\textbf{Via exhaustive search using simulation or an accurate performance analysis tool:} One naive way of solving the above mentioned network design problem is to consider all possible candidate tree topologies, and simulate each of them for a wide range of arrival rates to obtain their QoS respecting throughput regions, and choose the one with the largest throughput region. This method is clearly inefficient as simulation of each topology takes significant amount of time, and there could be exponentially many candidate trees (see Figure~\ref{fig:illustrative-topological-example}). An alternative approach is to replace the simulation step with a network analysis tool (such as the one proposed in \cite{srivastava} for IEEE~802.15.4 CSMA/CA networks) which is considerably faster compared to simulations; however, one still requires to evaluate an exponential number of candidate trees for a wide range of arrival rates, and hence the method is still inefficient.

\gap
\noindent
\textbf{Via a characterization of the QoS respecting throughput region:} A more efficient way of solving the network design problem would be to obtain an exact analytical characterization of the QoS respecting throughput region of a tree network under IEEE~802.15.4 CSMA/CA \emph{in terms of the topological properties of the network}, and then \emph{derive network design rules} from that characterization to maximize the throughput region. The difficulty with this approach is that for practical CSMA/CA protocols such as IEEE~802.15.4, obtaining an explicit exact characterization of the QoS respecting throughput region in terms of topological properties is notoriously hard. See Section~\ref{subsec:related-work} for more details. Therefore, some approximate methodology is in order.

\gap
\noindent
\textbf{Our strategy:} Our approach in solving this problem is two-fold:
\begin{enumerate}
\item Obtain an \emph{explicit} approximate inner bound to the QoS respecting throughput region of a tree network in terms of the topological properties of the network, and parameters of the CSMA/CA protocol. 
\item Obtain a tree that maximizes this approximate inner bound.
\end{enumerate}
Such an endeavour requires performance models of general multihop wireless networks under the CSMA/CA MAC, and the derivation of design criteria from such models. In this paper, we utilize our simplification of a detailed fixed point based analysis of a multi-hop tree network operating under IEEE~802.15.4 unslotted CSMA/CA \cite{IEEE}, provided by Srivastava et al. \cite{srivastava}, to develop certain explicit design criteria for QoS respecting networks.

As will be explained in a later section (see Section~\ref{sec:equal-arrival-rates}), it turns out that for default protocol parameters of IEEE~802.15.4 CSMA/CA, and for a wide range of QoS targets, the resulting solution is surprisingly simple: connect each sensor to the sink using a shortest path (in terms of hop count). Although this criterion is based on an approximate inner bound to the QoS respecting throughput region, we will see in our numerical experiments that this actually achieves larger throughput than a wide range of competing topologies. 

\begin{figure}[t]
\footnotesize
\begin{center}
\includegraphics[scale=0.35]{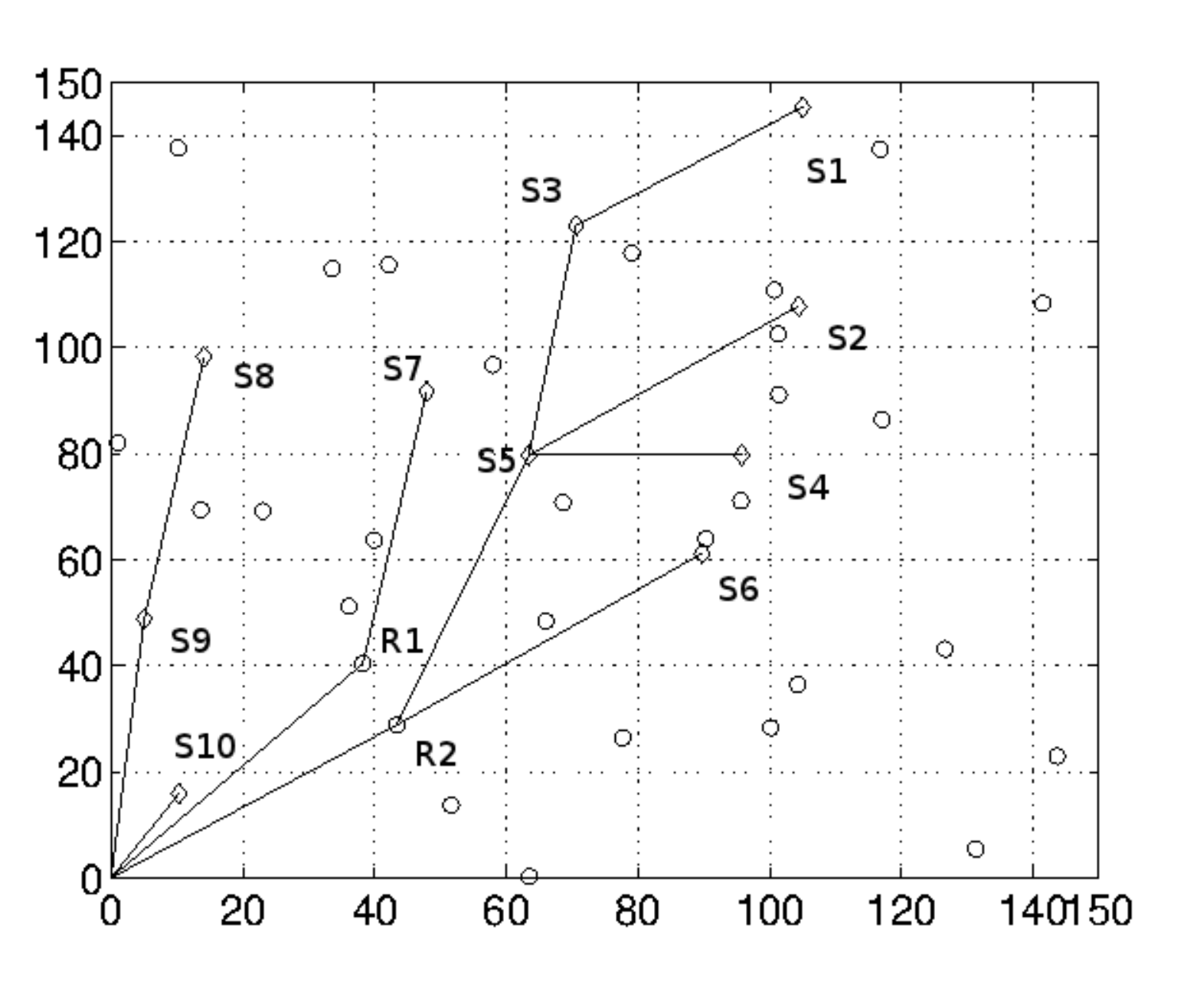}
\hspace{1mm}
\includegraphics[scale=0.36]{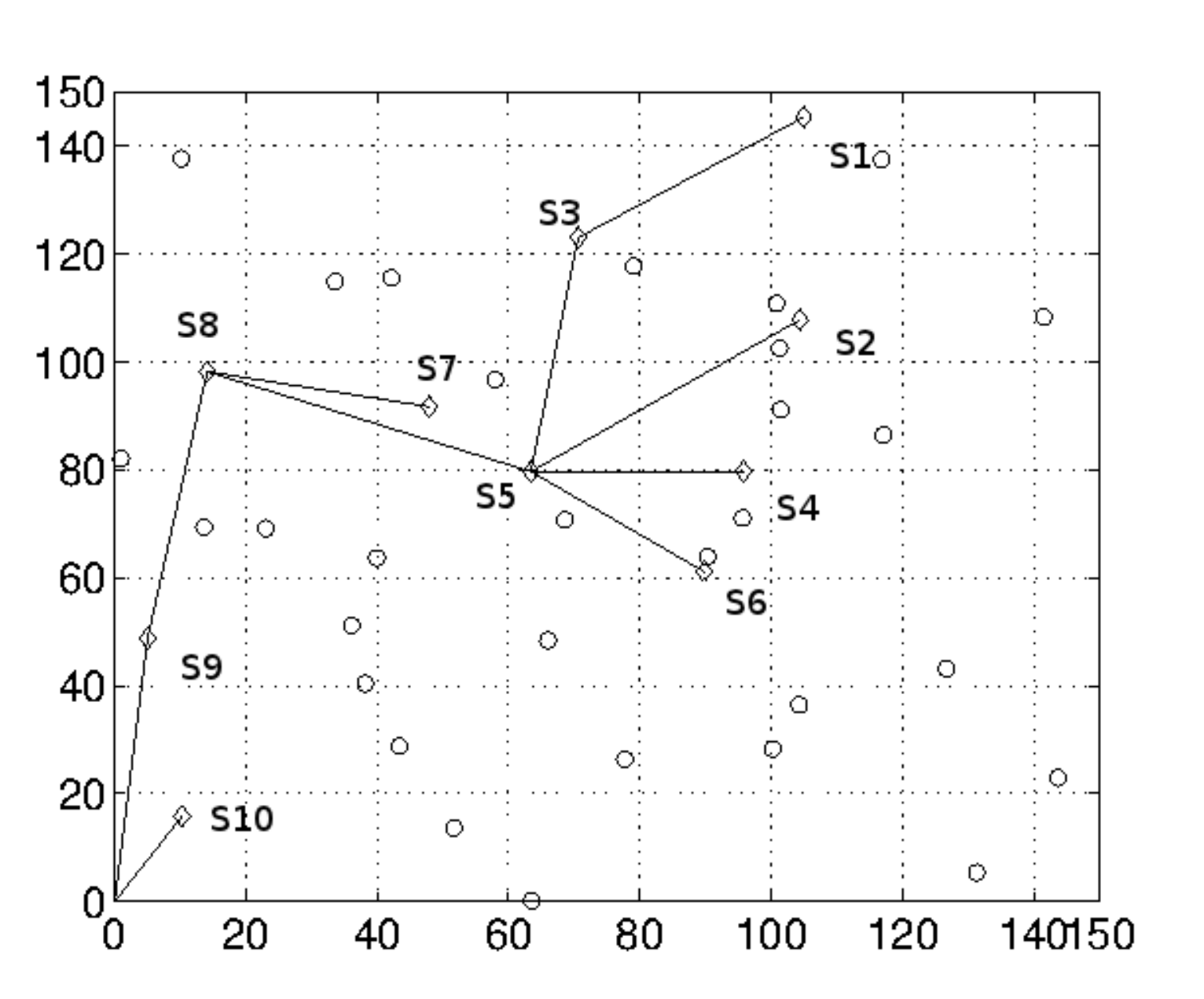}
\caption{Two competing tree topologies obtained from the network graph in Figure~\ref{fig:illustrative-topological-example}. Left panel: a shortest path tree connecting the sources to the BS. Note that it uses two relays, namely R1 and R2. Right panel: a tree obtained using an approximation algorithm for a certain Steiner graph design problem (see the SPTiRP algorithm \cite{bhattacharya-kumar14comnet}) to connect the sources to the sink. Note that it uses no relays, but has a higher total hop count compared to the shortest path tree.}
\label{fig:competing-topologies}
\vspace{-5mm}
\end{center}
\normalsize
\end{figure}
Continuing the example of Figure~\ref{fig:illustrative-topological-example}, Figure~\ref{fig:competing-topologies} demonstrates two different tree topologies connecting the sensors to the sink, that are both subgraphs of the example network graph of Figure~\ref{fig:illustrative-topological-example}. The left panel shows a shortest path tree connecting the sources to the BS; this tree requires two relays. The right panel depicts a tree obtained using the SPTiRP algorithm proposed in \cite{bhattacharya-kumar14comnet} for construction of hop constrained Steiner trees with a small number of relays. In the context of this example, our results in this paper provide

\gap
\noindent
i. explicit formulas for inner bounding the set of arrival rates that can be carried by either of these two topologies, while respecting the QoS objectives, and 

\noindent
ii. a basis for asserting that, for equal arrival rates from all sources, the shortest path tree topology in Figure~\ref{fig:competing-topologies} will achieve the larger QoS respecting arrival rate. 

Indeed, for equal arrival rates from all the sensors, it turns out (by brute force search over a wide range of arrival rates using the network analysis method presented in \cite{srivastava} for IEEE~802.15.4 CSMA/CA) that the shortest path tree can handle up to 5 packets/sec from each source, whereas the other tree topology in Figure~\ref{fig:competing-topologies} can handle 3.5 packets/sec from each source. The corresponding values from our explicit inner bound formulas are obtained as 3.511 packets/sec and 2.605 packets/sec respectively. See Section~\ref{sec:equal-arrival-rates} for details.   

\subsection{Related work} 
\label{subsec:related-work}
Most of the previous work on network design focuses on the ``lone-packet model,'' or a TDMA like MAC, where there is no contention in the network; see, for example, \cite{klein, Lin, ZhangHou, Misra, Costa, fullpaper, smartconnect-paper, multisink}, etc. On the other hand, most work on network performance modeling focuses on estimating some measure of network performance  under some simplified form of CSMA/CA MAC protocol for a \emph{given} arrival rate vector, and a \emph{given} network topology; see, for example, \cite{rachitpaper} and the references therein. This kind of work usually involves capturing the complex interaction (due to contention) among the nodes in the network via a set of fixed point equations (obtained by making some ``mean-field'' type \cite{bordenave} independence approximations), and then solving these equations using an iterative scheme to obtain the quantities of interest (see \cite{bianchi00performance, kumar-etal04new-insights} for example). For practical MAC protocols such as IEEE~802.11 and IEEE~802.15.4, these fixed point based schemes are too involved to gather any insight on the general dependence of the network performance on the arrival rates and network topology; see, for example, \cite{jindal09, park13, rachitpaper}. The only papers (see \cite{marbach11,bordenave}) that provide some insight on the topology dependence of the network performance do so for simplified MAC protocols that are not implemented in practice. In particular, Marbach et al. \cite{marbach11} provide an approximate inner bound on the stability region of a multi-hop network for a simplified CSMA/CA MAC, and Bordenave et al. \cite{bordenave} provide an approximate characterization of the network stability region for slotted ALOHA; however, neither work considers any QoS objective.

Another line of work studies the problem of joint rate control, routing and scheduling (\cite{jain-etal03interference, lin-shroff04rate-control, lin-shroff06imperfect-scheduling}) for throughput utility optimization in a given network graph. The scheduling algorithms proposed in \cite{jain-etal03interference} and \cite{lin-shroff04rate-control} are \emph{centralized}, and therefore, hard to implement in practice. In \cite{lin-shroff06imperfect-scheduling}, the authors propose a distributed rate control and scheduling algorithm (relatively easily implementable) that is guaranteed to achieve at least half the optimal throughput region of a given network. However, none of these papers consider any QoS objective. A related set of work focuses on developing \emph{queue length aware} CSMA/CA type distributed algorithms that can achieve the optimal throughput region, or some guaranteed fraction thereof, of a given network, while also yielding low delay \footnote{Note that these scheduling algorithms essentially allow infinite number of packet retransmissions to ensure 100\% end-to-end delivery}; see \cite{jiang-walrand12delay-distributed-scheduling} and the references therein for a detailed account of the progress in this line of research. However, these queue length based distributed algorithms are not implemented in any commercially available CSMA/CA (e.g., IEEE~802.11, IEEE~802.15.4). Therefore, we shall not pursue this line of work any further in this paper. 

Our current work is aimed as a first step towards bridging the gap between network performance modeling, and network design for a practical CSMA/CA MAC. In particular, we are interested in deriving, from network performance models, insights on how the network performance is affected by the topology and the arrival rate vector, and then convert that insight into useful criteria for network design. 

\subsection{Network model, QoS objectives, throughput regions, and some notation}
\label{subsec:nw-model-throuput-notation}
\subsubsection{Network model}
\label{subsubsec:setting}
In this paper we focus on developing criteria for the design of tree topologies with bounded path lengths; in the context of wireless sensor networks, such a problem arises when designing a QoS aware one-connected network spanning the sources and the base-station, placing additional relays if required, such that the number of hops from each sensor to the base-station is bounded by a given number, say, $h_{\max}$. 
 
We consider a tree network $T=(V, E_T)$ spanning a set of source nodes, $Q$ (including the BS), with $|Q|=m+1$, and a set of relays $R_T$ such that $V= Q\cup R_T$. We assume that the edges in $E_T$ are such that the packet error probability on any edge (for fixed packet lengths, non-adaptive modulation, fixed bit rates, and fixed transmit powers) is no more than a target value, $l$; additional relays may be needed to limit link lengths in order to achieve such a bound on packet error probabilities. We assume that the source-sink path lengths in $T$ obey a given hop constraint, $h_{\max}$. As explained in \cite{bhattacharya-kumar14tr-qos-aware-nw-design-csma}, the choice of such a hop constraint may be affected by several factors, including the RF propagation in the given environment, and the physical characteristics of the mote hardware. In this work, we shall not concern ourselves with details of the particular choice of $h_{\max}$. We further assume that the nodes use IEEE~802.15.4 CSMA/CA for medium access. Moreover, all the nodes are assumed to be in the carrier sense range of one another. We shall refer to this assumption as the \emph{No Hidden Nodes (NH)} assumption. 

\gap
\noindent
\textbf{Justification for the NH assumption}

While the NH assumption may seem restrictive, observe that by using a cluster based network topology where within each cluster, the nodes satisfy the NH assumption, and different channels are used across adjacent clusters (we recall that there are 16 channels available for IEEE~802.15.4 CSMA/CA \cite{IEEE}), moderate to large areas can be covered. One way of constructing such a cluster based topology is explained below. 

We assume that each cluster will cover a regular hexagonal region in the plane, and within each cluster, there will be a single BS placed at the center of the hexagon. Suppose the carrier sense range of each node is $r_{\mathrm{cs}}$, and, for a fixed packet size and a given path loss model, the maximum allowed link length to ensure the target link PER $l$ is $r$. Further suppose that $r$ is chosen such that $r_{\mathrm{cs}}=2mr$, for some integer $m > 0$. 

Having thus chosen $r$, \emph{we let the length of each side of a hexagonal cluster be $mr$}. Then, it follows from triangle inequality that all nodes within a hexagonal cluster are within CS range of each other.

\begin{figure}[h]
\begin{center}
\includegraphics[width=9cm]{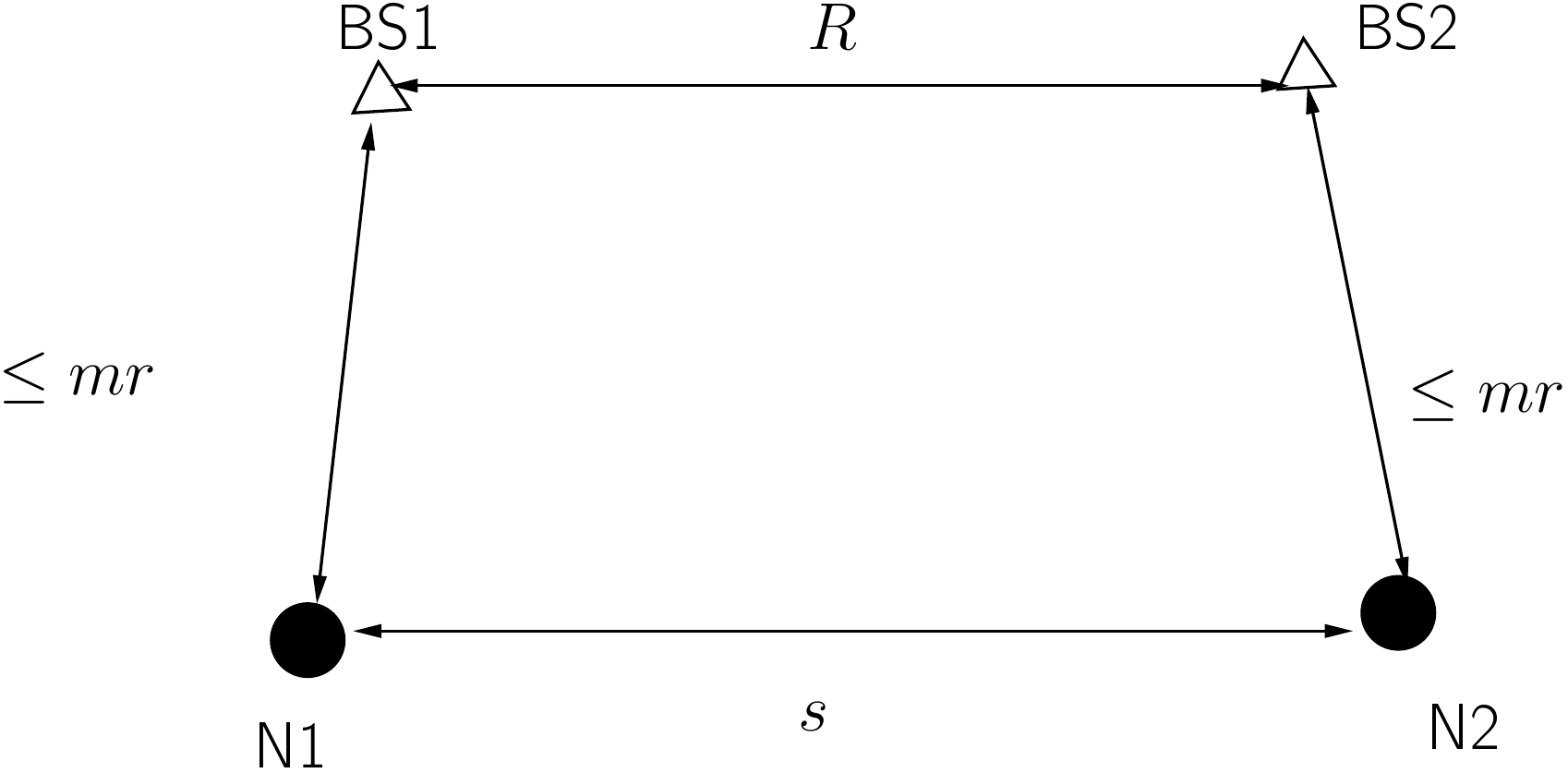}
\caption{Computation of the minimum distance between nodes in two different clusters}
\label{fig:dist_between_bs}
\vspace{-5mm}
\end{center}
\end{figure}
Now consider two distinct clusters. Let $R$ be the distance between the BSs in these two clusters. Consider a pair of nodes, one picked from each of these clusters. Let $s$ be the distance between these two nodes. See Figure~\ref{fig:dist_between_bs}. It follows from triangle inequality that
\begin{align}
R &\leq s + 2mr\nonumber\\
\Rightarrow s &\geq R - 2mr\nonumber
\end{align}
Thus, the distance between any pair of points across the two clusters is at least $R - 2mr$.

We can use the same channel in these two clusters if the nodes in one cluster are completely hidden from those in the other, so that they do not interfere with one another. This is ensured if $R - 2mr\geq r_{\mathrm{cs}}= 2mr$. Thus, it suffices to have $R\geq 4mr$. In other words, we can use the same channel in two different clusters if the distance between the BSs in these two clusters is at least $4mr$.

\begin{figure}[h]
\begin{center}
\includegraphics[width=9cm]{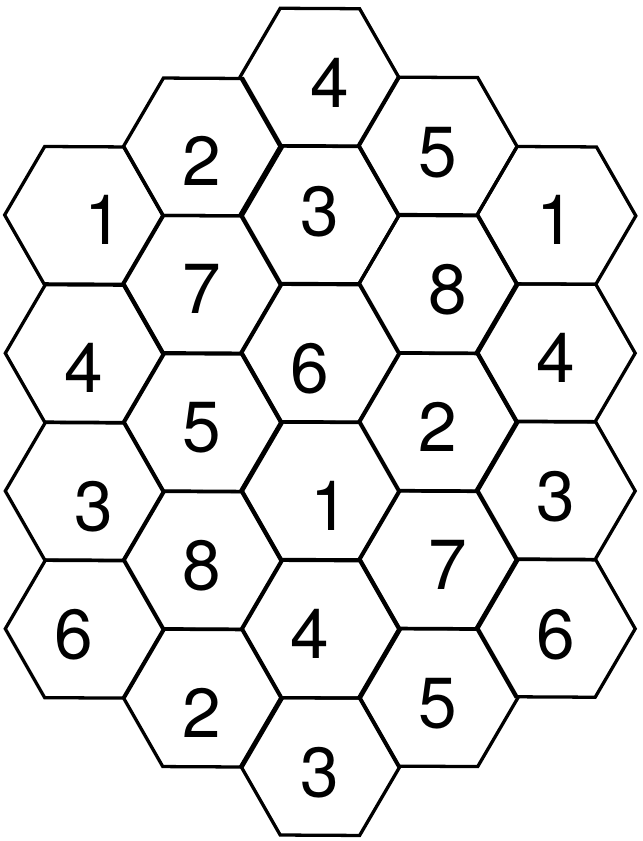}
\caption{Channel assignment to hexagonal clusters to ensure no interference across clusters}
\label{fig:hexagonal_nh}
\vspace{-5mm}
\end{center}
\end{figure}
Now consider the channel assignment in the hexagonal cluster layout shown in Figure~\ref{fig:hexagonal_nh}, where the length of each side of each hexagon is $mr$. It can be verified using elementary geometry that the BSs of any two co-channel clusters are separated by more than $4mr$, and hence with this channel assignment, there is no interference across clusters, while within each cluster, the NH assumption is satisfied. Note than we needed to use only 8 channels out of the available 16 channels. 

\subsubsection{QoS objectives}
\label{subsubsec:qos}

In order to design multi-hop networks that provide QoS, or even to route connections over such networks, given the intractability of analysis of multi-hop networks (particularly, CSMA/CA wireless networks), it has been the practice to adopt the approach of splitting the QoS (such as end-to-end mean delay, or delivery probability) over the hops along which a flow or connection is routed (see, for example, \cite[Section~5.10.1]{kmk1}). In this paper we adopt the approach that the end-to-end QoS objectives are \emph{split equally} over the links on each path. That is, to meet the target end-to-end delivery probability objective $p_{\mathrm{del}}$, we require that the packet discard probability on each link is no more than $\overline{\delta}\eqqcolon 1-\exp\left(\frac{\ln p_{\mathrm{del}}}{h_{\max}}\right)$. Similarly, to meet the target end-to-end mean delay objective $d_{\max}$, we require that the mean delay on each link is no more than $\overline{d}\eqqcolon\frac{d_{\max}}{h_{\max}}$.

Note that this splitting of QoS over the links of a path depends on the choice of $h_{\max}$. We discuss below, how the hop constraint, $h_{\max}$, can be chosen in practice. 

\gap
\noindent
\textbf{Choice of $h_{\max}$:} Recall that we define the lone-packet traffic model as one in which there is at most one packet traversing the network at any point in time. Under such a traffic model, there is no contention between links, but packets can still be lost due to channel errors and, therefore, packet discard can happen after a number of attempts (as defined by the protocol). It was proved in \cite{fullpaper} that when the probability of packet delivery is a measure of QoS, it is necessary to meet the QoS constraint under the lone-packet model in order to satisfy the QoS objective under any positive traffic arrival rate. 

Assuming the PER on each link to be the worst case link PER $l$, the mean delay on a link under the lone-packet model can be computed using an elementary analysis (see \cite{abhijitmeth}), taking into account the backoff behavior of IEEE~802.15.4 CSMA/CA, and using the backoff parameters given in the standard \cite{IEEE}. Then, to meet the mean delay requirement of $d_{\max}$ on a path with $h$ hops, we require that

\begin{equation}
h\leq \left\lfloor\frac{d_{\max}}{\overline{d}_{single-hop}}\right\rfloor\triangleq h_{\max}^{delay}
\end{equation}
where, $\overline{d}_{single-hop}$ is the mean link delay computed as explained earlier.

Again, assuming the worst case link PER of $l$, and the number of transmission attempts $n_t$ (obtained from the standard) before a packet is discarded on a link, the packet discard probability on a link \emph{under the lone-packet model} can be obtained as $\delta\:=\:l^{n_t+1}$. Hence, to ensure a packet delivery probability of at least $p_{\mathrm{del}}$ on a path with $h$ hops, we require that (assuming packet losses are independent across links)

\begin{eqnarray}
h &\leq \frac{\ln{p_{\mathrm{del}}}}{\ln{(1-\delta)}}&\triangleq h_{\max}^{delivery}
\end{eqnarray}

Hence, to ensure the QoS constraints under the lone-packet model, we require that the \emph{hop count on each path is upper bounded by $\min\{h_{\max}^{delay},h_{\max}^{delivery}\}$}.

Further, enforcing the \emph{no hidden nodes} assumption may also require us to constrain the hop length of each path, as discussed below. Suppose the carrier sense range of each node is $r_{\mathrm{cs}}$, and, for a fixed packet size and a given path loss model, the maximum allowed link length to ensure the target link PER $l$ is $r$. Then, if $h_{\max}$ is the maximum number of hops on any source-sink path, it follows by triangle inequality that to enforce the \emph{no hidden nodes} assumption, it suffices to have $2h_{\max}r\leq r_{\mathrm{cs}}\Rightarrow h_{\max}\leq \frac{r_{\mathrm{cs}}}{2r}\eqqcolon h_{\max}^{\mathrm{no-hidden}}$. 

Thus, if we define $\overline{h}_{\max}\eqqcolon \min\{h_{\max}^{delay},h_{\max}^{delivery},h_{\max}^{\mathrm{no-hidden}}\}$, then by choosing the hop constraint $h_{\max}$ such that $h_{\max}\leq \overline{h}_{\max}$, we can ensure that the end-to-end QoS objectives are met under the lone packet model, and also the \emph{no hidden nodes} condition is enforced. 

Note from the above discussion that the choice of $h_{\max}$ may depend on the RF propagation in the given environment, as well as the physical layer characteristics of the mote hardware. As we had mentioned earlier, we shall not concern ourselves with details of the particular choice of $h_{\max}$ in this work. With the values of $h_{\max}$, $p_{\mathrm{del}}$, $d_{\max}$ being given, we shall assume in the rest of the paper that we are given a target single-hop discard probability, $\deltabar$, and a target single-hop mean delay, $\dbar$.

\subsubsection{Throughput regions}
\label{subsubsec:throughput-def}
Suppose source $k$ generates traffic at rate $\lambda_k$ packets/sec, $k = 1,\ldots, m$, according to some ergodic point process. The throughput region of a given tree network, assuming that the nodes in the network use CSMA/CA for medium access, is defined as 
\begin{description}
\item $\Lambda(T)=\{\underline{\lambda}=\{\lambda_k\}_{k=1}^m:$ the network is \emph{stable}\}
\end{description}
Informally speaking, the network is said to be stable when the queue lengths do not grow unbounded with time.\footnote{See \cite{neely10stochastic-nw-optimization} for different notions of stability.} 

\cite{marbach11} and \cite{jindal09} attempted to characterize this throughput region for a simplified CSMA/CA, and for IEEE~802.11 CSMA/CA respectively for multi-hop networks. Note, however, that this notion of throughput region does not consider any QoS objective.

With our QoS splitting assumptions stated earlier, given a tree network $T$, and assuming the nodes in the network use CSMA/CA for medium access, we define 
\begin{description}
\item $\Lambda_{\deltabar}(T)=\{\underline{\lambda}:$ Under arrival rates $\underline{\lambda}$, the discard probability on each link in the tree $T$ is bounded by $\deltabar$\}
\item $\Lambda_{\dbar}(T)=\{\underline{\lambda}:$ Under arrival rates $\underline{\lambda}$, the mean delay on each link in the tree $T$ is bounded by $\dbar$\}
\item $\Lambda_{\deltabar,\dbar}(T) = \Lambda_{\deltabar}(T)\cap \Lambda_{\dbar}(T)$
\end{description}
Note that $\Lambda_{\deltabar,\dbar}\subset \Lambda$, for otherwise, at least one queue length would grow unbounded, and the mean delay constraint would not be met. 

For a tree $T$ with arrival rates $\lambda_k (packets/sec), 1 \leq k \leq m$, at the sources, we define 
\begin{description}
\item $\nu_i(\underline{\lambda},T):$ The packet rate into node $i$, if no packets are discarded in any node.
\item $h_k(T):$ The hop count on the path from source $k$ to the sink in tree $T$.
\end{description}  

\vspace{1mm}
For a given tree $T$, we then define the following throughput regions in terms of (i) a detailed fixed point analysis provided in \cite{srivastava}, (ii) a simplified version of this fixed point analysis that we provide in this paper, and (iii) numbers $B(\overline{\delta})$ and $B^{\prime}(\overline{\delta}, \overline{d})$ for which we will provide explicit formulas in terms of the parameters of the CSMA/CA protocol.

\begin{description}
\item $\hLambda_{\deltabar}(T)=\{\underline{\lambda}:$ Under arrival rates $\underline{\lambda}$, the discard probability on each link is bounded by $\deltabar$, \emph{according to the detailed analysis}\footnote{This means that if we analyze the tree $T$ under the given arrival rate vector $\underline{\lambda}$ using the detailed analysis, then the QoS values obtained from the analysis satisfy the target requirements.} presented in \cite{srivastava}\}

\item $\hLambda_{\dbar}(T)=\{\underline{\lambda}:$ Under arrival rates $\underline{\lambda}$, the mean delay on each link is bounded by $\dbar$, \emph{according to the detailed analysis} presented in \cite{srivastava}\}

\item $\hLambda_{\dbar,\deltabar}(T)=\hLambda_{\dbar}(T)\cap\hLambda_{\deltabar}(T)$

\item $\hhLambda_{\deltabar}(T)=\{\underline{\lambda}$: Under arrival rates $\underline{\lambda}$, the discard probability on each link is bounded by $\deltabar$, \emph{according to the simplified analysis} presented in Section~\ref{subsec:simplified-analysis}\}

\item $\hhLambda_{\dbar}(T)=\{\underline{\lambda}$: Under arrival rates $\underline{\lambda}$, the mean delay on each link is bounded by $\dbar$, \emph{according to the simplified analysis} presented in Section~\ref{subsec:simplified-analysis}\}

\item $\hhLambda_{\dbar,\deltabar}(T)=\hhLambda_{\dbar}(T)\cap\hhLambda_{\deltabar}(T)$

\item $\tLambda_{\deltabar}(T)= \{\underline{\lambda}:\sum_{k=1}^m\lambda_kh_k(T) < B(\deltabar)\}$  

\item $\tLambda_{\dbar,\deltabar}(T)=\{\underline{\lambda}:\sum_{k=1}^m\lambda_kh_k(T) < B(\deltabar) \text{ and }\max_{1\leq i\leq N}\nu_i \leq B^{\prime}(\deltabar,\dbar)\}$ 
\end{description}

\subsection{Contributions in this paper}
\label{subsec:contributions}
With the above definitions, the main contributions of this paper are summarized as follows: 

\begin{enumerate}
\item For the regime where packet discard probability is small, we obtain a simplified set of fixed point equations in Section~\ref{subsec:simplified-analysis}. Unlike \cite{srivastava}, \emph{we are able to show the uniqueness of the fixed point of our simplified equations} in Section~\ref{subsec:vector-uniqueness}.

\item Using the simplified fixed point analysis, we provide expressions for $B(\overline{\delta})$ and $B^{\prime}(\overline{\delta}, \overline{d})$ such that the following set relationships hold:

\begin{equation*}
\tLambda_{\dbar,\deltabar}(T)\stackrel{1}{\subset}\hhLambda_{\dbar,\deltabar}(T)\stackrel{2}{\approx}\hLambda_{\dbar,\deltabar}(T)
\end{equation*}

The set inequality~1 has been established in Sections~\ref{sec:arguments-leading-to-thm2} and \ref{sec:arguments-leading-to-thm3} (see Theorems~\ref{thm:membership-Lambda-delta} and \ref{thm:membership-lambda-dmax}), where, in the process, we provide expressions for the functions $B(\deltabar)$, and $B^{\prime}(\deltabar,\dbar)$. The ``tightness'' of Inequality~1 has been evaluated through numerical experiments (see Sections~\ref{subsubsec:equality-validation}, \ref{subsubsec:upper-bound-validation}, and Table~\ref{tbl:spt-throughput} in Section~\ref{subsubsec:comparison-against-sptirp}). Approximation~2 has been verified through extensive numerical experiments in Section~\ref{subsubsec:vector-fp-accuracy} for the regime where our simplified fixed point equations have a unique solution. 

These results are related to the throughput region of the original system, $\Lambda_{\deltabar,\dbar}(T)$, through the following approximation which was shown to be very accurate (well within 10\% in the regime where packet discard probability on a link is about the same as the PER on the link) in \cite{srivastava} by extensive comparison against simulations:

\begin{equation*}
\hLambda_{\dbar,\deltabar}(T)\approx\Lambda_{\deltabar,\dbar}(T)
\end{equation*}

\gap
\noindent
\remark It follows from the above set relationships that the explicitly defined set $\tLambda_{\dbar,\deltabar}(T)$ can be taken as an \emph{approximate inner bound} to the original throughput region, $\Lambda_{\deltabar,\dbar}(T)$. Furthermore, our numerical experiments suggest that $\tLambda_{\dbar,\deltabar}(T)$ captures a significant part of the throughput region $\hLambda_{\dbar,\deltabar}(T)$ (which, in turn, is a good approximation for $\Lambda_{\deltabar,\dbar}(T)$). In particular, when the QoS targets are in the range $\deltabar \geq 0.0209$, and $\dbar \geq 20$ msec, for the tested network topologies, it follows from our approximate inner bound that an arrival rate of at least 2-3 packets/sec from each source can be handled without violating QoS (see Table~\ref{tbl:spt-throughput} in Section~\ref{sec:numerical-results}, and the discussion thereafter). This arrival rate is more than enough for many sensor networking applications including those of industrial telemetry, and non-critical monitoring and control applications \cite{telemetry1,telemetry2}.

\item Finally, for the special case of equal arrival rates at all the sources, default backoff parameters of IEEE~802.15.4 CSMA/CA, and a range of target values $\deltabar$ and $\dbar$, we have shown in Section~\ref{sec:equal-arrival-rates} that 
\begin{align}
\tLambda_{\deltabar}(T)=\tLambda_{\dbar,\deltabar}(T)\label{membership:suff-condition-delta-suff-condition-delta-dbar}
\end{align}

Furthermore, let us define, for any tree $T$,
\begin{align}
\tilde{\lambda}(T) \eqqcolon \max_{\lambda: \lambda\mathbf{1}\in\tLambda_{\dbar,\deltabar}(T)}\lambda \nonumber
\end{align}
where, $\mathbf{1}$ is the vector of all 1's having the same length as the number of sources. 

Let $T^\ast$ be \emph{any shortest path tree spanning the sources and rooted at the sink}. Then, for any other tree $T$ spanning the sources and rooted at the sink, we argue in Section~\ref{sec:equal-arrival-rates} that
\begin{align}
\tilde{\lambda}(T)\leq \tilde{\lambda}(T^\ast)\label{eqn:spt-throughput-optimality}
\end{align}
\end{enumerate}

\section{Modeling CSMA/CA for a Network with No Hidden Nodes}
\subsection{The beaconless IEEE~802.15.4 CSMA/CA protocol\cite{srivastava}}
\label{subsec:csma-description}
When a node has data to send (i.e., has a non-empty queue), it initiates a random back-off with the first back-off period being sampled uniformly from $0$ to \emph{${2^{macminBE}-1}$}, where \emph{macminBE} is a parameter fixed by the standard. For each node, the back-off period is specified in terms of slots where a slot equals 20 symbol times ($T_{\mathsf{s}}$), and a symbol time equals 16 $\mu s$. The node then performs a CCA (\emph{Clear Channel Assessment}) to determine whether the channel is idle. We note that, unlike the IEEE 802.11 CSMA/CA, the back-off timer of a node is not ``frozen'' during the transmissions of other nodes. If the CCA succeeds, the node does a \emph{Rx-to-Tx turnaround}, which is $12$ symbol times, and starts transmitting on the channel. The failure of the CCA starts a new back-off process with the back-off exponent raised by one, i.e., to \emph{macminBE+1}, provided that this is less than its maximum value, \emph{macmaxBE}. The maximum number of successive CCA failures for the same packet is governed by \emph{macMaxCSMABackoffs}, exceeding which the packet is discarded at the MAC layer. The standard allows the inclusion of acknowledgements (ACKs) which are sent by the intended receivers on a successful packet reception. Once the packet is received, the receiver performs a \emph{Rx-to-Tx turnaround}, which is again $12$ symbol times, and sends a $22$ symbol fixed size ACK packet. A successful transmission is followed by an \emph{InterFrame Spacing}(IFS) before sending another packet.

When a transmitted packet collides or is corrupted by the PHY layer noise, the ACK packet is not generated, which is interpreted by the transmitter as failure in delivery. The node retransmits the same packet for a maximum of \emph{aMaxFrameRetries} times before discarding it at the MAC layer. After transmitting a packet, the node turns to the Rx-mode and waits for the ACK.  The \emph{macAckWaitDuration} determines the maximum amount of time a node must wait for in order to receive the ACK before concluding that the packet (or the ACK) has collided. The default values of \emph{macminBE}, \emph{macmaxBE}, \emph{macMaxCSMABackoffs}, and \emph{aMaxFrameRetries} are 3, 5, 4, and 3 respectively.

\subsection{Detailed fixed point equations for NH case}
\label{subsec:full-analysis}
Although there has been considerable research on analytical modeling of CSMA/CA, the work reported in \cite{rachitpaper}, \cite{srivastava}, in the context of beacon-less IEEE~802.15.4 CSMA/CA, appears to be the most comprehensive (capturing aspects such as multi-hopping, hop-by-hop queueing, presence of acknowledgements, packet collisions due to presence of hidden terminals, and dilation of perceived channel activity period due to hidden terminals, etc.), and also very accurate. We base our work in this paper on a simplified version of the detailed fixed point analysis reported in \cite{rachitpaper,srivastava}. For ease of reference, we briefly describe the modeling philosophy of their analysis here, and extract from that analysis the main equations for the NH case.

\subsubsection{The overall modeling approach (\cite{rachitpaper}, \cite{srivastava})}

In \cite{rachitpaper, srivastava}, akin to the approach in \cite{bianchi00performance} and \cite{kumar-etal04new-insights}, a ``decoupling'' approximation is made, whereby each node is modeled separately, incorporating the influence of the other nodes in the network by their average statistics, and as if these nodes were independent of the tagged node. In the literature, such an approach has also been called a ``mean field approximation,'' and formal justification has been provided in, for example, \cite{bordenave}.

\gap
\noindent
\emph{Modeling the activity of a tagged node, say $i$:} All packets entering node $i$ are assumed to have the same fixed length, and hence they take the same amount of time when being transmitted over the medium (denoted by $T_{\mathsf{tx}}$). In a network with no hidden nodes, the channel activity perceived by node $i$ is due to all the other nodes in the network. Each node alternates between periods during which its transmitter queue is empty and those during which the queue is non-empty. During periods in which its transmitter queue is non-empty, when the node is not transmitting it is performing repeated backoffs for the head-of-the-line (HOL) packet in its transmitter queue. Each backoff ends in a CCA attempt. Figure~\ref{fig:ARP_calc_of_rates} is a depiction of this alternation in node behavior during the periods when its queue is nonempty. We employ the decoupling approximation to analyse this process.
\begin{figure}[h]
\begin{center}
\includegraphics[width=9cm]{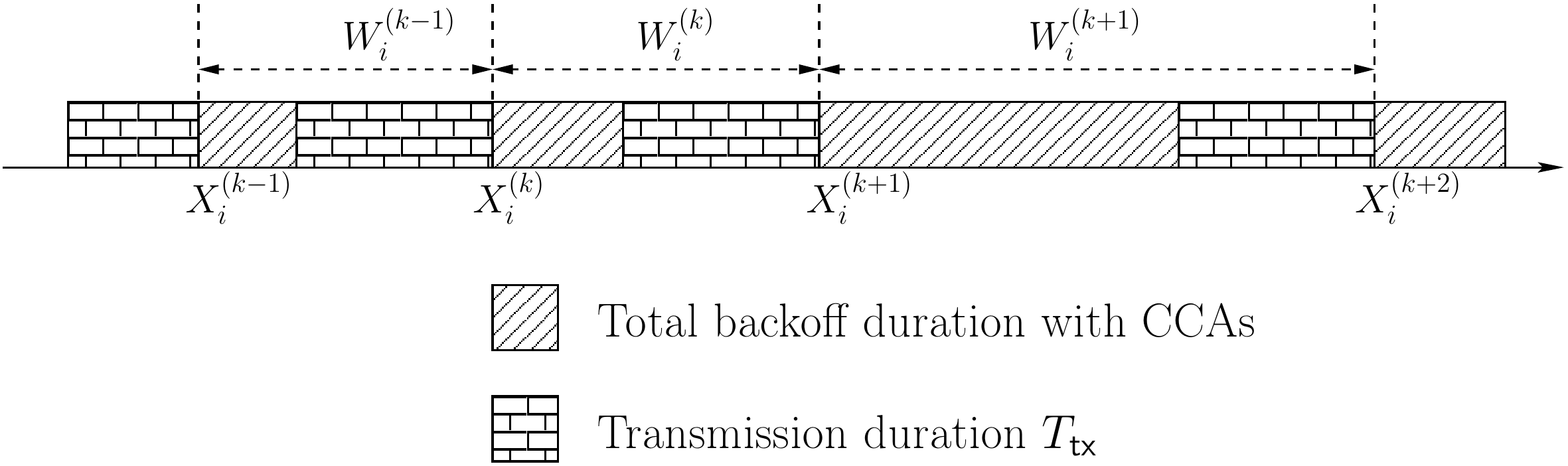}
\caption{(Taken from \cite{srivastava}) The alternating (contention-transmission) process obtained by observing the process at a node $i$ after removing all the idle time periods from the original process. The transmission completion epochs are denoted by $\{X_i^{(k)}\}$ and the contention-transmission cycle lengths by $\{W_i^{(k)}\}$.}
\label{fig:ARP_calc_of_rates}
\vspace{-5mm}
\end{center}
\end{figure}

Focusing on durations during which Node~$i$'s queue is nonempty, consider the completion of a transmission by Node~$i$; this could have been a success or a collision. Packet collisions can occur even in the absence of hidden nodes since transmissions from two transmitters placed within the CS range of each other can overlap (known as \textbf{Simultaneous Channel Sensing}) as depicted in Figure~\ref{fig:simul_channel}. 
\begin{figure}[h]
\begin{center}
	\includegraphics[scale=0.45]{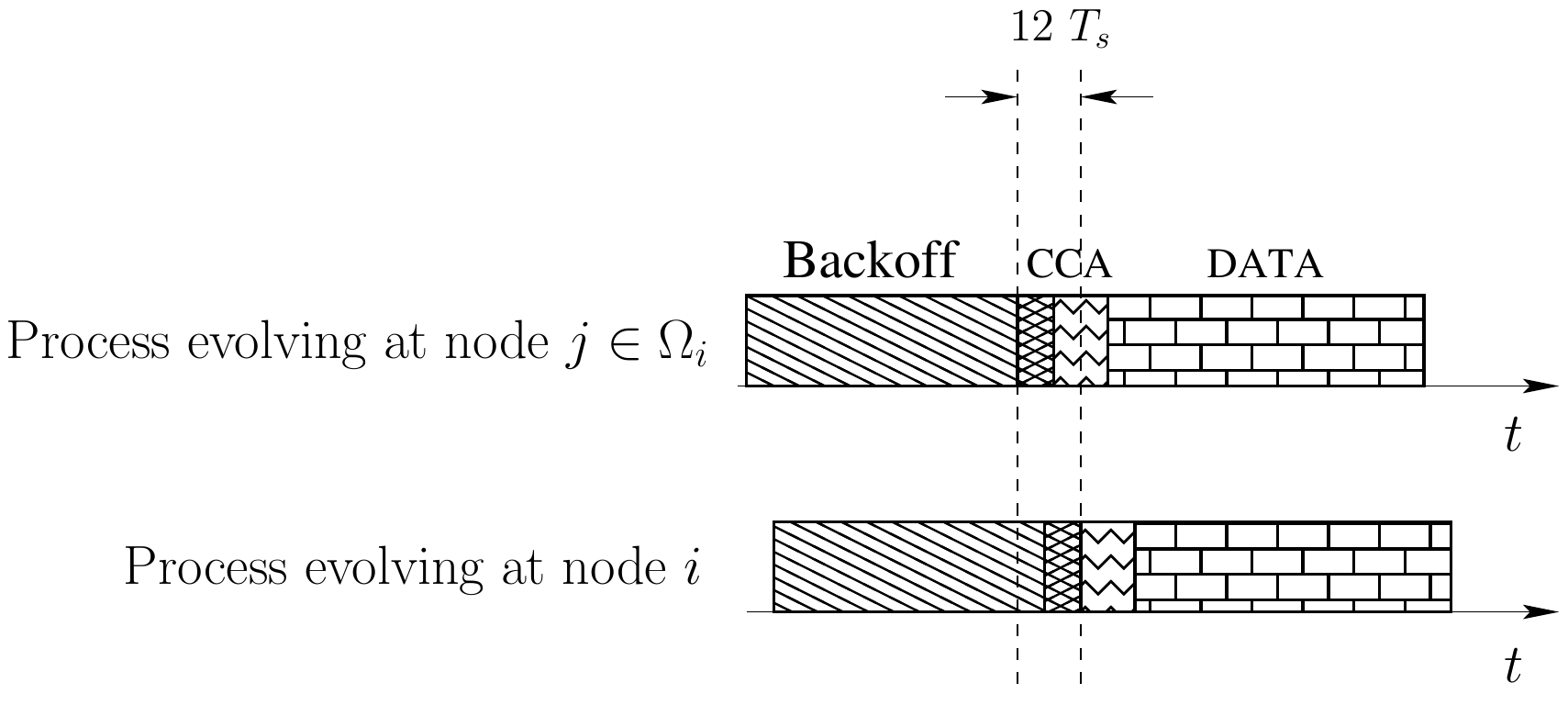}
	\caption{(Taken from \cite{srivastava}) Node $j$ finishes its backoff, performs a CCA, finds the channel idle and starts transmitting the DATA packet. Node $i$ finishes its backoff anywhere in the shown $12~T_s$ duration and as there is no other ongoing transmission in the network, its CCA succeeds and it enters the transmission duration. As a result, the DATA packets may collide at receiver of $i$ or that of $j$.}
	\label{fig:simul_channel}
	\vspace{-5mm}
\end{center}
\end{figure}

Due to the small propagation delay, and the no hidden nodes assumption, we neglect the up to 12 additional symbol durations for which other simultaneous transmissions could last. After the completion of its transmission, the node $i$ contends for the channel by executing successive backoffs; the nodes $j \neq i$ either have non-empty queues and are also contending, or have empty queues and are not contending. With these observations in mind we define the following CCA attempt rates. The CCA attempt process at node $i$ \emph{conditioned on being in backoff periods} is modeled as a Poisson process of rate $\beta_i$. For each node $j\ne i$, the CCA attempt process, conditioned on $j$'s queue being non-empty and being in backoff, or being empty, is modeled by an independent Poisson process with rate $\overline{\tau}_j^{(i)}$, $j\ne i$. This is the rate of attempts of Node~$j$ as perceived by Node~$i$.  

Now let us return to the process depicted in Figure~\ref{fig:ARP_calc_of_rates}, i.e., the alternating contention and transmission periods of Node~$i$ when its queue is nonempty. As a result of the assumption of independent Poisson attempt processes at the nodes, it can be observed that the instants $X^{(k)}_i$ are renewal instants, and the intervals $W_i^{(k)}$ form an
i.i.d.\footnote{Independently and identically distributed} sequence. 

With the above observations and definitions, the analysis in \cite{srivastava} now proceeds as follows. During a contention period of Node~$i$, the system alternates between nodes counting down their backoff counters, and some node other than $i$ transmitting. If a node other than $i$ transmits, then CCA attempts by $i$ during these times lead to $i$'s CCA failures. When a CCA does succeed, a collision can occur by the CCA of another node being performed within 12 symbol times of the first successful CCA. Given $\beta_i$ and $\tau_j^{(i)}$, and the independent Poisson processes assumption, these details can be analysed, leading to equations relating the collision probability, and the CCA failure probability to the assumed rates of CCA attempts. Moreover, the CCA attempt rate conditioned on being in backoff periods, $\beta_i$, can also be expressed in terms of the CCA failure probability of Node~$i$ using renewal theoretic arguments similar to those in \cite{kumar-etal04new-insights}. Given the collision probabilities, and the CCA failure probabilities, quantities such as the packet discard probability at a node, queue non-empty probability of a node, and the probability that a node is in back-off conditioned on its queue being non-empty can be analysed. Finally, having obtained these probabilities, the CCA attempt rates, $\tau_j^{(i)}$, can be expressed in terms of these probabilities using renewal theoretic arguments. The resulting fixed point equations, for the NH case, are shown in the next section.

The analysis of packet delay between a source node and the sink node utilizes the approximation techniques in Whitt's Queueing Network Analyzer (QNA, \cite{whitt83queueing-network-analyzer}). The arrival processes at the source nodes are modeled as independent Poisson processes, with given first two moments of the interarrival times. We recall that our network is a tree, rooted at the sink. Thus, each interior node receives packets from one or more upstream nodes, and feeds exactly one downstream node. Each node is modeled by a GI/GI/1 queue. The fixed point equations obtained as explained earlier are iterated until a convergence tolerance is met. The resulting values of the attempt rates, collision rates, CCA failure rates, etc., can be used to obtain the first two moments of the time spent by an HOL packet at a queue, i.e., the service time at the queue. Given the first two moments of the interarrival time and of the service time, QNA uses a standard GI/GI/1 mean delay approximation for the mean delay at a queue. QNA also provides a method for approximately modeling the superposition of the output processes of two or more GI/GI/1 queues by a renewal process, and thus provides approximations for the first two moments of the interarrival times into downstream queues. In this manner, starting with the model of the input processes at the source nodes, the end-to-end mean delay can be approximated. 

In \cite{srivastava}, the authors find that the analysis is accurate to within 10\% in the regime where the discard probability on a link is about the same as the PER on that link.

\subsubsection{Main equations}

Assuming that the CCA attempt process at each node $i$, conditioned on the node being in backoff, is distributed as $\exp(\beta_i)$, the CCA failure probability, $\alpha_i$, of node $i$ is given by
\begin{eqnarray}
\alpha_i &=& \frac{(1-\eta_i)(1-c_i)\beta_i T_{tx}}{\eta_i + (1-\eta_i)c_i + (1-\eta_i)(1-c_i)\beta_i T_{tx}}
\label{eqn:nh-alpha}
\end{eqnarray}
where, $T_{tx}$ is the packet transmission duration, $\eta_i$ is the probability that node $i$ finishes its backoff first, and $c_i$ is the probability that node $i$ finishes its backoff within 12 symbol times (the turnaround time) after another node finishes its backoff. This equation can be interpreted as a consequence of the Renewal-Reward Theorem (RRT)(e.g., see \cite{ross07prob-models}) as follows: the numerator is the mean number of failed CCA attempts by node $i$ in a renewal cycle (see Figure~\ref{fig:ARP_calc_of_rates}), and the denominator is the mean number of total CCA attempts by node $i$ in a renewal cycle. 

\begin{align*}
\eta_i &= \frac{\beta_i}{\beta_i + \displaystyle{\sum_{j\ne i}\overline{\tau}_j^{(i)}}} \\
c_i &= \Bigg(1-e^{-12\beta_i}\Bigg) 
\end{align*}
where, $\overline{\tau}_j^{(i)}$ is the CCA attempt rate of node $j$ \emph{conditioned on the times during which it is either empty or in backoff}, and is given by 
\begin{align}
\overline{\tau}_j^{(i)}&= \frac{\beta_j \times b_j \times q_j}{1 - q_j + q_j \times b_j}\label{eqn:nh-1}
\end{align}
where, $b_j$ is fraction of time node $j$ is in backoff given that it is non-empty, and $q_j$ is the probability that node $j$ is non-empty. The above expression can also be interpreted as a consequence of the RRT: the numerator is the mean number of CCA attempts by node $j$ in a renewal cycle, and the denominator is the mean time during which the node $j$ is not transmitting in a renewal cycle (both obtained after normalizing the renewal cycle length to unity). Also, note that at the start of the renewal cycle, node $j$ can be either empty, or in backoff (i.e., it cannot be transmitting); hence $\overline{\tau}_j^{(i)}$ is also the CCA attempt rate of node $j$ \emph{as perceived by the tagged node $i$}. We also have,
\begin{eqnarray}
b_i=\frac{\overline{B}_i}{\overline{B}_i + (1-\alpha_i^{n_c})T_{\mathsf{tx}}}
\label{eqn:nh-calc_of_b}
\end{eqnarray}
where, $\overline{B}_i$ is the mean time spent in backoff by the HOL packet at node $i$ before it is transmitted, or gets discarded due to successive CCA failures. $\overline{B}_i$ is a function of the backoff parameters of the protocol, and the maximum allowed successive CCA failures, $n_c$, before a packet is discarded. 

\begin{eqnarray}
q_i=\min\{1,\frac{\nu_i}{\sigma_i}\}
\label{eqn:nh-q_arrival}
\end{eqnarray}
where, $\nu_i$ is the total arrival rate into node $i$, and $\frac{1}{\sigma_i}$ is the mean service time of the HOL packet at node $i$. 

The packet failure probability, $\gamma_i$, at node $i$ is given by
\begin{eqnarray}
\gamma_i &=& p_i +(1-p_i )l
\label{eqn:nh-gamma}
\end{eqnarray}
where, $l$ is the PER on the outgoing link from node $i$, and $p_i$ is the collision probability for a packet transmitted from node $i$ (i.e., already conditioned on the fact that it did not encounter a CCA failure), and is given by
\begin{eqnarray}
p_i &=& \frac{R_i^{(3)}+R_i^{(4)}}{\eta_i+(1-\eta_i)c_i}									
\label{eqn:nh_p_hidden}
\end{eqnarray}
where,
\begin{eqnarray*}
\scriptstyle{R_i^{(3)}} &\scriptstyle{=}& \scriptstyle{\eta_i\Bigg(1-\exp \Bigg\{-12\Bigg( \displaystyle{\sum_{j\ne i}}\overline{\tau}_j^{(i)} \Bigg) \Bigg\}\Bigg)}
\end{eqnarray*}
is the unconditional probability that node $i$ finishes its backoff first, and then some other node finishes its backoff within the vulnerable period of 12 symbol times, thereby causing a collision.
 
\begin{eqnarray*}
R_i^{(4)} &=& (1-\eta_i)c_i 
\end{eqnarray*}
is the unconditional probability that some other node finished its backoff first, and then node $i$ finished its backoff within the vulnerable period of the aforementioned node, thus resulting in a collision.

$r_i = \gamma_i(1-\alpha_i^{n_c})$ denotes the probability that the HOL packet at node $i$ was transmitted (i.e., it did not encounter $n_c$ successive CCA failures), and it encountered a transmission failure (either due to collision or due to link error).
 
Then, the mean service time of the HOL packet at node $i$ is given by
\begin{eqnarray}
\frac{1}{\sigma_i}=\overline{Z}_i + \overline{Y}_i
\label{eqn:nh-sigma}
\end{eqnarray}
where,
\begin{align}
\overline{Z}_i &=\overline{B}_i(1+r_i+r_i^2+\ldots +r_i^{n_t-1})\label{eqn:nh-z}\\
\overline{Y}_i &=(1-\alpha_i^{n_c})T_{\mathsf{tx}}(1+r_i+r_i^2+\ldots +r_i^{n_t-1})\label{eqn:nh-y}
\end{align}
denote respectively, the mean time spent in backoff by the HOL packet, and the mean time spent in transmission by the HOL packet. Here $n_t$ is the number of transmission failures after which the packet is discarded. 

The packet discard probability at node $i$, denoted $\delta_i$, is given by
\begin{align}
\label{eqn:nh-delta}
\delta_i&= \alpha_i^{n_c}(1+r_i+r_i^2+\ldots +r_i^{n_t-1})+r_i^{n_t}
\end{align}

Then, the total arrival rate into node $i$ is given by
\begin{eqnarray}
\nu_i=\lambda_i+\sum_{k\in\mathcal{P}_i}\theta_k
\label{eqn:nh-nu}
\end{eqnarray}
where, 
\begin{eqnarray}
\theta_k=\nu_k(1-\delta_k)
\label{eqn:nh-theta_arrival}
\end{eqnarray}
is the goodput of node $k$. 

Finally, we also have
\begin{eqnarray}
\beta_i=\frac{1+\alpha_i+\alpha_i^2+\ldots +\alpha_i^{n_c-1}}{\overline{B}_i}
\label{eqn:nh-beta}
\end{eqnarray}
This can again be interpreted as a consequence of the RRT, where the numerator is the mean number of CCA attempts for the HOL packet at node $i$, and the denominator is the mean time spent in backoff by the HOL packet at node $i$. 
 
Suppose that the given tree (with $N$ nodes) satisfies the hop constraint $h_{\max}$ on the path from each source to the sink. Suppose $L_j$ is the set of nodes on the path from source $j$ to the sink. Then, assuming packet discards are independent across nodes, \emph{the end-to-end delivery probability} of source $j$ is given by $\prod_{i\in L_j}(1-\delta_i)$, where $\delta_i$ is the packet discard probability at node $i$ as defined earlier.

Another set of approximations were proposed in \cite{srivastava} based on Whitt's Queueing Network Analyzer \cite{whitt83queueing-network-analyzer} to approximate the end-to-end delay from a source to the sink. The end-to-end mean packet delay for a source node $j$, provided that the set of nodes along the path from this node to the BS is $L_j$, is given by
\begin{eqnarray}
\Delta_j &=& \displaystyle{\sum_{i\in L_j}\overline{\Delta}_i}
\end{eqnarray}
where, $\overline{\Delta}_i$ is the  mean sojourn time at a node $i$, and is given by

\begin{eqnarray}
\overline{\Delta}_i &=& \frac{\rho_i\mathbb{E}(S_i)(c_{A_i}^2 + c_{S_i}^2)}{2(1-\rho_i)} + \mathbb{E}(S_i)
\end{eqnarray}
Here, $\rho_i$ denotes the traffic load at node $i$, $S_i$ denotes the service time at node $i$, $c_{S_i}^2$ is the squared coefficient of variance of service time at node $i$, and $c_{A_i}^2$ is the squared coefficient of variance for the interarrival times at node $i$. We also have,

\begin{eqnarray*}
c_{S_i}^2 &=& \frac{\mathbb{E}(S_i^2)}{(\mathbb{E}(S_i))^2} - 1\\
\mathbb{E}(S_i) &=& -\frac{d(M_{S_i}(z))}{dz}\Bigg \vert_{z=0} \\
\mathbb{E}(S_i^2) &=& \frac{d^2(M_{S_i}(z))}{dz^2}\Bigg \vert_{z=0}
\end{eqnarray*}
where, $M_{S_i}(z)$ is the MGF of the service time of node $i$, and can be expressed as,
\begin{eqnarray*}
M_{S_i}(z) &=& \frac{\beta_i(1-\alpha_i)(1-\gamma_i)e^{-zT_{\mathsf{tx}}}}{z+\beta_i(1-\alpha_i)(1-\gamma_ie^{-zT_{\mathsf{tx}}})}.
\end{eqnarray*}
with $\beta_i$, $\alpha_i$, $\gamma_i$ and $T_{\mathsf{tx}}$ having the same interpretation as in Section~\ref{subsec:full-analysis}. 

Further, $c_{A_i}^2$ is calculated as 
\begin{eqnarray*}
c_{A_i}^2 &=& \frac{1}{\Lambda_i}\Bigg(\lambda_i + \displaystyle{\sum_{j\in \mathcal{P}_i}}\Lambda_jc_{D_j}^2\Bigg)
\end{eqnarray*}
where, $\Lambda_i$ is the total arrival rate into Node~$i$, $\lambda_i$ is the exogeneous Poisson arrival rate into Node~$i$ ($\lambda_i=0$ if Node~$i$ is a relay node), $\mathcal{P}_i$ is the set of predecessors of node $i$, and 
\begin{eqnarray*}
c_{D_i}^2 &=& (1-\delta_i)(1 + \rho_i^2(c_{S_i}^2 - 1) + (1-\rho_i^2)(c_{A_i}^2 - 1))
\end{eqnarray*}

It is evident from the complex structure of the fixed point equations displayed above that it is difficult to use this analysis to extract insight about the relation between network topology and the QoS measures. We therefore, in the next section, resort to simplifying these detailed equations in a certain regime of network operation.

\subsection{Simplified fixed point equations in the low discard regime}
\label{subsec:simplified-analysis}
For practical purposes, we are primarily interested in modeling the behavior of CSMA/CA for arrival rates at which the discard probability is low. We now proceed to further simplify the above fixed point equations in this low discard regime by making the following approximations:

\begin{description}
\item [(A1)] $\delta_i\approx 0$ so that, by Eqn.~\eqref{eqn:nh-nu}, the total arrival rate into node $j$ is $\nu_j\approx \sum_{k=1}^m z_{k,j}\lambda_k$, where $z_{k,j}=1$ iff node $j$ is in the path of source $k$, and $z_{k,j}=0$ otherwise. $\lambda_k$ is the exogeneous arrival rate at source $k$. 
\item [(A2)] $r_i\ll 1$, i.e., the probability that the HOL packet encounters a transmission failure is very small.
\item [(A3)] $q_i\ll 1$, i.e., the fraction of time that the queue of node $i$ is non-empty is very small.
\item [(A4)] $c_i \approx 0$, i.e., the probability that node $i$ finishes its backoff within 12 symbol times after another node has finished backoff is negligibly small.
\item [(A5)] $\gamma_i \ll 1$, so that $(1-\gamma_i)\approx 1$, i.e., the probability that a transmitted packet encounters a collision or link error is very small.
\end{description}

Let us define
\begin{align*}
\overline{\tau}_{-i}&\eqqcolon \sum_{j\neq i}\overline{\tau}_j^{(i)}
\end{align*}
where $\overline{\tau}_j^{(i)}$ is given by Equation~\eqref{eqn:nh-1}. Using Equation~\eqref{eqn:nh-1}, we have, for all $i=1,\ldots,N$,

\begin{align}
\overline{\tau}_{-i}&= \sum_{j\neq i}\frac{\beta_j \times b_j \times q_j}{1 - q_j + q_j \times b_j}\nonumber\\
&\approx \sum_{j\neq i}\beta_j \times b_j \times q_j,\:\text{using (A3)}\nonumber\\
&= \sum_{j\neq i} \frac{1 + \alpha_j + \ldots+\alpha_j^{n_c-1}}{\overline{B}_j + (1-\alpha_j^{n_c})T_{\mathsf{tx}}}\nu_j(\overline{Z}_j + \overline{Y}_j)\label{eqn:fp1}\\
&\approx \sum_{j\neq i} \nu_j(1 + \alpha_j + \ldots+\alpha_j^{n_c-1})\nonumber\\
&\approx \sum_{j\neq i}\Bigg(\sum_{k=1}^m z_{k,j}\lambda_k\Bigg)(1 + \alpha_j + \ldots+\alpha_j^{n_c-1}),\:\text{using (A1)}\nonumber
\end{align}
In writing Equation~\eqref{eqn:fp1}, we have used Equations~\eqref{eqn:nh-beta}, \eqref{eqn:nh-calc_of_b}, \eqref{eqn:nh-q_arrival}, and \eqref{eqn:nh-sigma}. In writing the last step, we have used (A2) along with eqns.~\ref{eqn:nh-z} and \ref{eqn:nh-y} to approximate $\overline{Z}_j + \overline{Y}_j\approx \overline{B}_j + (1-\alpha_j^{n_c})T_{\mathsf{tx}}$.

Also, from Equation~\eqref{eqn:nh-alpha}, using (A4), we have, for all $i=1,\ldots,N$,
\begin{align*}
\alpha_i&\approx \frac{T_{\mathsf{tx}}\overline{\tau}_{-i}}{1 + T_{\mathsf{tx}}\overline{\tau}_{-i}}
\end{align*}

Thus, we can write the simplified fixed point equations compactly as follows:
\begin{align}
\overline{\tau}_{-i}&=\sum_{j\neq i}\Bigg(\sum_{k=1}^m z_{k,j}\lambda_k\Bigg)(1 + \alpha_j + \ldots+\alpha_j^{n_c-1})\:\forall i=1,\ldots,N\label{eqn:simple-tau-i}\\
\alpha_i &= \frac{T_{\mathsf{tx}}\overline{\tau}_{-i}}{1 + T_{\mathsf{tx}}\overline{\tau}_{-i}}\:\forall i = 1,\ldots,N\label{eqn:simple-alpha}
\end{align}
Equation~\ref{eqn:simple-tau-i} can be interpreted as the total CCA attempt rate seen by node $i$ due to the other nodes; indeed, the first term inside the summation is the packet arrival rate into node $j$, while the second term is the mean number of CCA attempts of a packet at node $j$.
   
Further, using (A4), Equation~\eqref{eqn:nh-gamma} simplifies to
\begin{align}
\gamma_i &= l + (1-l)(1-\exp(-12\overline{\tau}_{-i}))\label{eqn:simple-gamma}
\end{align}

\gap
\noindent
\textbf{Simplifications for the delay approximation}

Recall that in the low discard regime, $\delta_i\approx 0$, and $\rho_i\ll 1$ so that $\rho_i^2\approx 0$. Then it is easy to see that $c_{A_i}^2\approx 1$. Furthermore, $\rho_i\approx \nu_i \mathbb{E}(S_i)$, where we recall that $\nu_i$ is the total arrival rate into node $i$ assuming no packet discard. Thus, we have

\begin{align}
\overline{\Delta}_i &= \frac{\rho_i\mathbb{E}(S_i)(1 + c_{S_i}^2)}{2(1-\rho_i)} + \mathbb{E}(S_i)
\end{align}

Note that the above expression is precisely the Pollaczek-Khintchine formula (see, for example, \cite{ross07prob-models}) for the mean delay of an M/G/1 queue. Thus, in the low discard regime, the delay approximation reduces to modeling each node in the network by an M/G/1 queue.

Starting with the MGF $M_{S_i}(z)$, and taking derivative w.r.t $z$, straightforward calculations yield,
\begin{align}
\mathbb{E}(S_i)&=-\frac{d}{dz}M_{S_i}(z)\Bigg\vert_{z=0}\nonumber\\
&= T_{\mathrm{tx}}+\frac{1+\beta_i(1-\alpha_i)\gamma_iT_{\mathrm{tx}}}{\beta_i(1-\alpha_i)(1-\gamma_i)}\nonumber\\
&= \frac{1+\beta_i(1-\alpha_i)T_{\mathrm{tx}}}{\beta_i(1-\alpha_i)(1-\gamma_i)}\label{eqn:mean-serv-rate-delay}
\end{align}

Now, let us compute $c_{S_i}^2$. Again, taking the derivative of $M_{S_i}(z)$ w.r.t $z$, straightforward computations yield,
\begin{align}
\mathbb{E}(S_i^2)&= \frac{d^2}{dz^2}M_{S_i}(z)\Bigg\vert_{z=0}\nonumber\\
&= T_{\mathrm{tx}}^2 + \frac{2T_{\mathrm{tx}}}{\beta_i(1-\alpha_i)(1-\gamma_i)}\nonumber\\
&+\frac{3\beta_i(1-\alpha_i)\gamma_iT_{\mathrm{tx}}^2}{\beta_i(1-\alpha_i)(1-\gamma_i)}\nonumber\\
&+ \frac{2(1+\beta_i(1-\alpha_i)\gamma_iT_{\mathrm{tx}})^2}{(\beta_i(1-\alpha_i)(1-\gamma_i))^2}\label{eqn:mean-of-squared-serv-time}
\end{align}

Finally, after some algebraic manipulations, we have
\begin{align}
c_{S_i}^2 &= \frac{Var(S_i)}{\mathbb{E}^2(S_i)}\nonumber\\
&= \gamma_i + \frac{1-\gamma_i}{(1+\beta_i(1-\alpha_i)T_{\mathrm{tx}})^2}\nonumber\\
&\approx \gamma_i + \frac{1}{(1+\beta_i(1-\alpha_i)T_{\mathrm{tx}})^2}\label{eqn:coeff-var-serv-time}
\end{align}
where, in Eqn.~\eqref{eqn:coeff-var-serv-time}, we have used (A5). 

\gap
\noindent
\remark
In our numerical experiments with default backoff parameters, the above simplifications yield delay values that are accurate to within 6\% w.r.t the delay values obtained from the detailed analysis. 

In what follows, we shall use this simplified fixed point analysis to obtain sufficient conditions for the membership of $\hLambda_{\overline{\delta}}$ and $\hLambda_{\dbar}$. 

We start by establishing a condition for the uniqueness of the solution to the simplified vector fixed point equations~\eqref{eqn:simple-tau-i} and \eqref{eqn:simple-alpha}.

\subsection{Existence and uniqueness of the simplified vector fixed point}
\label{subsec:vector-uniqueness}

Let us define $\alpha_{\max}\eqqcolon \overline{\delta}^{1/n_c}$, and $a\eqqcolon \frac{\alpha_{\max}}{T_{tx}(1-\alpha_{\max})}$, where $\overline{\delta}$ is the target discard probability (at a node) as defined earlier. Observe from Equation~\eqref{eqn:simple-alpha} that $\alpha_i\leq \alpha_{\max}\Leftrightarrow \overline{\tau}_{-i}\leq a$. Moreover, note that in our regime of interest, i.e., the regime where $\max_{1\leq i\leq N}\delta_i\leq \overline{\delta}$, we have $\alpha_i\leq \alpha_{\max}$ for all $i=1,\ldots,N$, or equivalently, $\overline{\tau}_{-i}\leq a$ for all $i=1,\ldots,N$. Then, we have the following proposition. 

\begin{proposition}
\label{prop:vector-uniqueness}
If 
\begin{align}
\sum_{k=1}^m \lambda_kh_k &< \min\big\{\frac{a}{1+\alpha_{\max}+\ldots\alpha_{\max}^{n_c-1}},\nonumber\\
&\frac{1}{T_{tx}(1+2\alpha_{\max}+\ldots +(n_c-1)\alpha_{\max}^{n_c-2})}\Bigg\}\eqqcolon B_1(\overline{\delta})\label{eqn:uniqueness-condition}
\end{align}
then the simplified fixed point equations defined by \eqref{eqn:simple-tau-i} and \eqref{eqn:simple-alpha} have a unique solution $\{\overline{\tau}_{-i}\}_{i=1}^N$ in $[0,a]^N$. 
\end{proposition}

\begin{proof}
The proof proceeds via a series of lemmas that establish that the simplified fixed point equations define a contraction mapping on $[0,a]^N$. See Appendix~\ref{subsec:proof-vector-uniqueness} for details. 
\end{proof}

We call the regime defined by Equation~\eqref{eqn:uniqueness-condition} the \emph{uniqueness regime}. Our numerical experiments suggest that in this regime, the simplified fixed point equations well-approximate the original fixed point equations (to within 10\% for $\overline{\tau}_{-i}$), i.e., we have $\hLambda_{\cdot}\approx \hhLambda_{\cdot}$, where $\hLambda_{\cdot}$ and $\hhLambda_{\cdot}$ are as defined in Section~\ref{subsubsec:throughput-def}. See Section~\ref{sec:numerical-results} for details. 

\section{Derivation of A Sufficient Condition for Membership of $\hhLambda_{\deltabar}$}
\label{sec:arguments-leading-to-thm2}

In this section, we shall derive a sufficient condition for the membership of the set $\hhLambda_{\deltabar}(T)$ (defined in Section~\ref{subsubsec:throughput-def}) for a given tree $T$. In other words, we shall derive the structure of the set $\tLambda_{\deltabar}(T)$ introduced in Section~\ref{subsubsec:throughput-def}. 

\subsection{A control on the maximum packet discard probability}
From equation~\eqref{eqn:nh-delta}, we can derive the following Lemma.

\begin{lemma}
\label{lem:delta-inc-alpha-gamma}
For $\alpha_i^{n_c} \leq 1/n_t$, and $\gamma_i^{n_t}\leq 1/n_t$, the packet discard probability at node $i$, $\delta_i$, is monotonically increasing in $\alpha_i$ and $\gamma_i$, i.e., $\delta_{i,1}\leq \delta_{i,2}$ whenever $\alpha_{i,1}\leq \alpha_{i,2}$ and $\gamma_{i,1}\leq \gamma_{i,2}$.
\end{lemma}

\begin{proof}
See Appendix~\ref{subsec:proof-delta}.
\end{proof}

From the simplified fixed point equations~\eqref{eqn:simple-tau-i}-\eqref{eqn:simple-gamma}, we have the following lemma. Since the proof is short, we present the proof here itself. 

\begin{lemma}
\label{lem:alpha-gamma-inc-tau}
$\alpha_i$ and $\gamma_i$ are monotonically increasing in $\overline{\tau}_{-i}$.
\end{lemma}

\begin{proof}
From Equation~\eqref{eqn:simple-gamma}, it is clear that $\gamma_i$ is monotonically increasing in $\overline{\tau}_{-i}$. To see that $\alpha_i$ is monotonically increasing in $\overline{\tau}_{-i}$, observe that the derivative of the R.H.S of Equation~\eqref{eqn:simple-alpha} w.r.t $\overline{\tau}_{-i}$ is non-negative. 
\end{proof}

Combining Lemmas~\ref{lem:delta-inc-alpha-gamma} and \ref{lem:alpha-gamma-inc-tau}, we have 

\begin{proposition}
\label{prop:delta-inc-tau}
$\delta_i$ is monotonically increasing in $\overline{\tau}_{-i}$.
\end{proposition}

It follows from Proposition~\ref{prop:delta-inc-tau} that to control the maximum packet discard probability, we need to control $\displaystyle{\max_{i=1,\ldots,N}}\overline{\tau}_{-i}$, i.e.,

\begin{align*}
\max_{i=1,\ldots,N}\delta_i\leq \overline{\delta}\Leftrightarrow \max_{i=1,\ldots,N}\overline{\tau}_{-i}\leq \tau_{\max}
\end{align*}
where $\tau_{\max}$ is the upper bound on $\overline{\tau}_{-i}$ that ensures $\delta_i\leq \overline{\delta}$.

\subsection{A scalar fixed point}
\label{subsec:scalar-fp}
Before proceeding further, we take a slight detour, and introduce a further simplification to the fixed point equations described in Sections~\ref{subsec:full-analysis} and \ref{subsec:simplified-analysis}. While Equation~\eqref{eqn:uniqueness-condition} gives us a condition for uniqueness of the simplified vector fixed point in terms of topological properties of the network, it does not provide any explicit information as to the dependence of the vector fixed point (or, more importantly, the network performance) on \emph{topological properties} within the uniqueness regime. We therefore, aim to further simplify the fixed point equations to obtain a scalar fixed point equation which will be exploited later on to extract information about topological dependencies. 

We start with the following Lemma. 
\begin{lemma}
\label{lem:tau}
(i) For all $j$, and for any given positive\footnote{by ``positive'', we mean $\lambda_k > 0$ for all $k=1,\ldots,m$} arrival rate vector $\{\lambda_k\}_{k=1}^m$, $\overline{\tau}_j^{(i)}$ is bounded away from zero.

\noindent
(ii) For all $j$, $\overline{\tau}_j^{(i)}$ is uniformly upper bounded by a constant that depends only on the backoff parameters of the protocol. 
\end{lemma}
\begin{proof}
See Appendix~\ref{subsec:proof-lemma-tau}.
\end{proof}

From the Lemma~\ref{lem:tau}, we have the following insight:
\begin{insight}
\label{cor:equality}
In the large $N$ regime, for all $i=1,\ldots,N$,
\begin{equation*}
\overline{\tau}_{-i}\approx \sum_{j=1}^N\overline{\tau}_j^{(i)}\eqqcolon \overline{\tau}
\end{equation*}
i.e., the total attempt rate seen by a node is roughly equal for all nodes, and is approximately equal to the total attempt rate in the network.
\end{insight}

The argument behind this insight is as follows: it follows from Part (i) of Lemma~\ref{lem:tau} that $\displaystyle{\lim_{N\rightarrow \infty}\sum_{j=1}^N\overline{\tau}_j^{(i)}}=\infty$. But from Part (ii) of Lemma~\ref{lem:tau}, we have that for each $j=1,\ldots, N$, $\overline{\tau}_j^{(i)}$ is uniformly upper bounded by a finite constant, independent of $N$. Hence, in the large $N$ regime, we can conclude that for all $j$, $\overline{\tau}_j^{(i)}\ll \sum_{j=1}^N\overline{\tau}_j^{(i)}$. Thus, in this regime, $\overline{\tau}_{-i}\eqqcolon\sum_{j\neq i}\overline{\tau}_j^{(i)}\approx \sum_{j=1}^N\overline{\tau}_j^{(i)}\eqqcolon \overline{\tau}$. 

The above insight, along with a slight modification of the derivations leading to the vector fixed point equations~\eqref{eqn:simple-tau-i},\eqref{eqn:simple-alpha}, suggests the following scalar fixed point equation in $\overline{\tau}$, the total attempt rate in the network:

\begin{align}
\overline{\tau}&= (\sum_{k=1}^m\lambda_kh_k)(1+\alpha+\ldots+\alpha^{n_c-1})\label{eqn:scalar-tau}\\
\alpha &= \frac{T_{tx}\overline{\tau}}{1 + T_{tx}\overline{\tau}}\label{eqn:scalar-alpha}
\end{align}

\subsection{Existence and uniqueness of the scalar fixed point}
\label{subsec:scalar-uniqueness}

Let us, as before, define $a\eqqcolon \frac{\alpha_{\max}}{T_{tx}(1-\alpha_{\max})}$, where $\alpha_{\max}=\overline{\delta}^{1/n_c}$. Observe from Eqn.~\eqref{eqn:scalar-alpha} that $\alpha\leq \alpha_{\max}\Leftrightarrow \overline{\tau}\leq a$. Then, we have the following proposition.

\begin{proposition}
\label{prop:scalar-uniqueness}
Suppose the condition defined by Eqn.~\eqref{eqn:uniqueness-condition} holds. Then, the scalar fixed point equation defined by \eqref{eqn:scalar-tau},\eqref{eqn:scalar-alpha} has a unique solution in $[0,a]$. 
\end{proposition}

\begin{proof}
The proof proceeds by showing that the scalar fixed point equation defines a contraction mapping on $[0,a]$. See Appendix~\ref{subsec:proof-scalar-uniqueness} for details.  
\end{proof}

\subsection{A tight upper bound on $\max_{1\leq i \leq N}\overline{\tau}_{-i}$}
We now come back to the main thread of our discussion. Recall the scalar fixed point equations~\eqref{eqn:scalar-tau}, \eqref{eqn:scalar-alpha} derived in Section~\ref{subsec:scalar-fp}. We had promised that these equations will find use in deriving topological properties that control the network performance. To this end, we start with the following proposition.

\begin{proposition}
\label{prop:scalar-vector-fp-relation}
In the uniqueness regime (defined by Eqn.~\ref{eqn:uniqueness-condition}), 
\begin{align}
\overline{\tau}\geq \max_{1\leq i\leq N}\overline{\tau}_{-i}\label{eqn:scalar-vector-fp-relation}
\end{align}
where, $\overline{\tau}$ is the unique solution to the scalar fixed point equations~\eqref{eqn:scalar-tau},\eqref{eqn:scalar-alpha}, and $\{\overline{\tau}_{-i}\}_{i=1}^N$ is the unique solution to the simplified vector fixed point equations~\eqref{eqn:simple-tau-i},\eqref{eqn:simple-alpha}.
\end{proposition} 

\begin{proof}
See Appendix~\ref{subsec:proof-scalar-vector}.
\end{proof}

\gap
\noindent
\emph{Remarks:}

\begin{enumerate}
\item The intuition behind the above proposition is clear. Each component of the simplified vector fixed point approximates the total attempt rate \emph{seen} by the corresponding node in the network, which should clearly be upper bounded by the total attempt rate of all the nodes in the network, approximated by the scalar fixed point. 
\item Corollary~\ref{cor:equality} suggests that the above upper bound on $\max_{1\leq i\leq N}\overline{\tau}_{-i}$ is tight in the large $N$ regime. Our numerical experiments show that the bound is tight even for moderate values of $N$ (see Section~\ref{sec:numerical-results} for details). 
\end{enumerate}

From Proposition~\ref{prop:scalar-vector-fp-relation}, it follows that to ensure $\max_{1\leq i\leq N}\overline{\tau}_{-i}\leq \tau_{\max}$, it \emph{suffices} that we ensure $\overline{\tau}\leq \tau_{\max}$, i.e., to control the maximum packet discard probability, $\max_{1\leq i\leq N}\delta_i$ (or, equivalently, $\max_{1\leq i\leq N}\overline{\tau}_{-i}$; recall Proposition~\ref{prop:delta-inc-tau}, and the discussion thereafter), it suffices to control the approximate total attempt rate, $\overline{\tau}$. We, therefore, next investigate the dependence of $\overline{\tau}$ on the network topology and arrival rate vector to come up with a sufficient condition for the membership of $\hhLambda_{\overline{\delta}}$.

\subsection{A sufficient condition for membership of $\hhLambda_{\overline{\delta}}$}
\label{subsec:sufficient-cond-Lambda-delta}
The following proposition is an easy consequence of Part~2 of Lemma~\ref{lem:scalar-uniqueness-lipschitz} stated in Appendix~\ref{subsec:proof-scalar-uniqueness}, and hence we omit the proof.

\begin{proposition}
\label{prop:tau-inc-total-load}
The approximate total attempt rate, $\overline{\tau}$, given by the solution to the scalar fixed point equations~\eqref{eqn:scalar-tau},\eqref{eqn:scalar-alpha}, is monotonically increasing in $\sum_{k=1}^m\lambda_kh_k$, the total load in the network.
\end{proposition}

It follows from Proposition~\ref{prop:tau-inc-total-load} that given $\overline{\delta}$ (or, equivalently, $\tau_{\max}$), there exists $B_2(\overline{\delta})$ such that $\sum_{k=1}^m\lambda_kh_k\leq B_2(\overline{\delta})\Leftrightarrow \overline{\tau}\leq \tau_{\max}$. Hence, as long as the vector fixed point equations~\eqref{eqn:simple-tau-i},\eqref{eqn:simple-alpha} and the scalar fixed point equations~\eqref{eqn:scalar-tau},\eqref{eqn:scalar-alpha} have unique solutions, the following cause-effect relations hold:
\begin{equation*}
\sum_{k=1}^m\lambda_kh_k\leq B_2(\overline{\delta})\Rightarrow \overline{\tau}\leq \tau_{\max}\Rightarrow \max_{1\leq i\leq N}\overline{\tau}_{-i}\leq \tau_{\max}\Rightarrow \max_{1\leq i\leq N}\delta_i\leq \overline{\delta}
\end{equation*}

However, Eqn.~\eqref{eqn:uniqueness-condition} gives us a sufficient condition for the uniqueness of the fixed point equations. Thus, we have the following theorem:

\begin{theorem}
\label{thm:membership-Lambda-delta}
Let $B(\overline{\delta})\eqqcolon \min\{B_1, B_2\}$, where $B_1$ is as defined in Eqn.~\eqref{eqn:uniqueness-condition}, and $B_2$ is as defined in the above discussion. Suppose, for a given tree network topology, an arrival rate vector $\mathbf{\lambda}$ satisfies 
\begin{align}
\sum_{k=1}^m \lambda_kh_k < B(\overline{\delta})\label{eqn:pdel-membership}
\end{align}
where, $m$ is the number of sources, and $h_k$ is the hop count on the path from source $k$ to the sink. Then, $\mathbf{\lambda}\in \hhLambda_{\overline{\delta}}$.
\end{theorem}

It follows that an arrival rate vector $\{\lambda_k\}_{k=1}^m\in \tLambda_{\deltabar}(T)$ for a given tree $T$ \emph{iff} it satisfies Equation~\eqref{eqn:pdel-membership}.

\section{Derivation of a Sufficient Condition for the Membership of $\hhLambda_{\deltabar}\cap\hhLambda_{\dbar}$}
\label{sec:arguments-leading-to-thm3}

In this section, we shall derive the structure of the set $\tLambda_{\dbar,\deltabar}(T)$ for a given tree $T$. 

\subsection{Dependence of $\mathbb{E}(S_i)$ and $c_{S_i}^2$ on $\overline{\tau}_{-i}$}

We make the following claim. Since the proof is short, we provide the proof here itself. 

\begin{lemma}
\label{lem:service-rate-inc-tau}
$\mathbb{E}(S_i)$ and $c_{S_i}^2$ are monotonically increasing in $\overline{\tau}_{-i}$.
\end{lemma}
\begin{proof}
It is easy to see from Eqn.~\eqref{eqn:mean-serv-rate-delay} that $\mathbb{E}(S_i)$ is coordinatewise monotonically increasing in $\alpha_i$ and $\gamma_i$ (since, $\beta_i$ is monotonically decreasing in $\alpha_i$). Then, the claim follows from Lemma~\ref{lem:alpha-gamma-inc-tau}. 

The argument for $c_{S_i}^2$ is identical, starting with Eqn.~\eqref{eqn:coeff-var-serv-time}.
\end{proof}

\subsection{A bound on $\overline{\Delta}_i$}
Recall from Proposition~\ref{prop:delta-inc-tau} that the constraint on the packet discard probability, namely $\max_{1\leq i\leq N}\delta_i\leq \overline{\delta}$, translates to a constraint on $\overline{\tau}_{-i}$, namely $\max_{1\leq i\leq N}\overline{\tau}_{-i}\leq \tau_{\max}$. This, together with Lemma~\ref{lem:service-rate-inc-tau}, implies a bound on $\mathbb{E}(S_i)$ and on $c_{S_i}^2$, namely that, for all $i$, $\mathbb{E}(S_i)\leq \overline{S}$, and $c_{S_i}^2\leq \overline{c_S^2}$. Noting that $\frac{\rho_i}{1-\rho_i}$ is monotonically increasing in $\rho_i$, This yields a bound on $\overline{\Delta}_i$ for fixed $\nu_i$, namely,

\begin{align}
\overline{\Delta}_i\leq \frac{\nu_i\overline{S}}{2(1-\nu_i\overline{S})}\overline{S}(1+\overline{c_S^2})+\overline{S}\nonumber
\end{align}

\subsection{A sufficient condition for membership of $\hhLambda_{\dbar}\cap \hhLambda_{\deltabar}$}

In order that $\overline{\Delta}_i\leq \overline{d}$ for all $i$, it suffices that for all $i$, $\frac{\nu_i\overline{S}}{2(1-\nu_i\overline{S})}\overline{S}(1+\overline{c_S^2})+\overline{S}\leq \overline{d}$. After some straightforward algebraic manipulations, this yields an upper bound, designated by $B^{\prime}(\overline{\delta},\overline{d})$, on $\nu_i$ for all $i$. Thus, we have the following theorem:
\begin{theorem}
\label{thm:membership-lambda-dmax}
If an arrival rate vector $\mathbf{\lambda}=\{\lambda_k\}_{k=1}^m$ satisfies Equation~\eqref{eqn:pdel-membership}, and the following holds:
\begin{align}
\max_{1\leq i\leq N}\nu_i \leq B^{\prime}(\overline{\delta},\overline{d})\label{eqn:dmax-membership}
\end{align}
then, $\mathbf{\lambda}\in\hhLambda_{\overline{\delta}}\cap\hhLambda_{\overline{d}}$. 
\end{theorem}

Note that from the definition of $\tLambda_{\deltabar}(T)$, and Theorem~\ref{thm:membership-lambda-dmax}, it also follows that $\tLambda_{\deltabar}(T)\supset \tLambda_{\dbar,\deltabar}(T)$. 

\section{Special Case: Equal Arrival Rates at All Sources}
\label{sec:equal-arrival-rates}

In this section, we shall focus on the notion of equal throughput, i.e., the scenario where the external arrival rates at all the sources are equal, say, $\lambda$. This leads to the following interesting consequences. 

\begin{enumerate}
\item Condition~\eqref{eqn:dmax-membership} reduces to $\lambda\max_{1\leq i\leq N}m_i\leq B^{\prime}$, whereas, condition~\eqref{eqn:pdel-membership} reduces to $\lambda\sum_{i=1}^N m_i \leq B$, where $m_i$ is the number of sources using node $i$. 
Now, for default backoff parameters, and for $\overline{\delta}=.0208$, $\overline{d}=20$ msec (corresponding to, for example, $\pdel=90\%$ and $\dmax=100$ msec for $h_{\max}=5$), it turns out that $B(\overline{\delta})=$  80.75 packets/sec, and $B^{\prime}(\overline{\delta},\overline{d})=$ 82.85 packets/sec, i.e., $B(\overline{\delta}) < B^{\prime}(\overline{\delta},\overline{d})$ so that for all $\lambda$ satisfying the discard probability objective~\eqref{eqn:pdel-membership}, $\lambda\max_{1\leq i\leq N}m_i\leq \lambda\sum_{i=1}^N m_i < B(\overline{\delta}) < B^{\prime}(\overline{\delta},\overline{d})$, i.e., the mean delay objective~\eqref{eqn:dmax-membership} is also satisified. Hence, it follows that for equal arrival rates, and default backoff parameters, $\tLambda_{\dbar,\deltabar}=\tLambda_{\deltabar}$. Furthermore, it was observed that this conclusion continues to hold for a range of QoS targets such that $.015 \leq \deltabar \leq 0.04$ and $20\text{ msec}\leq \dbar \leq 45$ msec.

\item \textbf{A simple network design criterion:} Suppose we are given a graph $G=(V,E)$ with $V=Q\cup R$, where $R$ is a set of \emph{potential} locations where one can place relays, and $E$ is the set of \emph{admissible} edges with PER at most $l$. Our objective is to design a tree network spanning the sources and the BS, possibly using a few relays, such that the resulting network meets a given hop constraint $h_{\max}$ on each source-sink path, and meets given per-hop discard probability and mean delay targets, while achieving a large throughput region. 

In light of the discussion in Item 1 above, for default backoff parameters and reasonable QoS targets, an approximate lower bound to the throughput of a given tree network topology is $\frac{B(\overline{\delta})}{\sum_{k=1}^m h_k}$ (using $\tLambda_{\dbar,\deltabar}$ as an approximate inner bound to $\Lambda_{\deltabar,\dbar}$; recall the set relationships from Section~\ref{subsec:contributions}). This, in turn, gives a \emph{simple criterion for throughput optimal network design} for the case of equal arrival rates, namely, \emph{obtain a tree that minimizes the total hop count from the sources to the sink, $\sum_{k=1}^m h_k$}; this is nothing but the \emph{shortest path tree} (with hop count as cost).\footnote{Note that if the hop constraint, $h_{\max}$, is feasible, the shortest path tree will also meet the hop constraint.} Note that this is simply a re-wording of Equation~\eqref{eqn:spt-throughput-optimality} stated in Section~\ref{subsec:contributions}.

Observe that the same criterion holds for the notion of max min throughput (see, for example, \cite{bhattacharya-kumar14tr-qos-aware-nw-design-csma} for a formal definition), since the max-min is achieved along the direction of equal arrival rates. 
\end{enumerate}

As we shall see in our numerical experiments (Section~\ref{sec:numerical-results}), although the shortest path tree criterion is based on an approximation to $\Lambda_{\deltabar,\dbar}$, it, in fact, achieves better throughput than a wide range of competing topologies.  

\section{Numerical Results}
\label{sec:numerical-results}

\subsection{The setting}
\label{subsec:experiment-setting}
We conducted all our numerical experiments using the default protocol parameters of IEEE~802.15.4 CSMA/CA. In particular, we assumed $n_c = 5$, $n_t = 4$, and default back-off parameters (\cite{IEEE}) in all the experiments. 

We chose packet error rate (PER) on each link to be $l=2\%$, and packet length $T=131$ bytes\footnote{This includes 70 bytes of data, 8 bytes of UDP header, 20 bytes IP header, 27 bytes MAC header, and 6 bytes Phy header.}. The target packet discard probability, $\overline{\delta}$, was chosen to be 0.0208, and the target single-hop mean delay was chosen to be 20 msec (which correspond, for example, to a target end-to-end delivery probability $p_{\mathrm{del}}=90\%$, and target end-to-end mean delay $d_{\max}=100$ msec for a hop constraint $h_{\max}=5$). We assumed equal arrival rates at all the sources for the experiments. 

All the experiments were conducted on 5 random network topologies, each of which was generated as follows: 10 sources and 30 potential relay locations were selected uniformly at random over a 150$\times$150 $m^2$ area. The sink was chosen at a corner point of the area. A Steiner tree was then formed connecting the sources to the sink, using no more than 4 relays\footnote{This restriction was imposed since, in practice, the network planner would like to budget the use of additional relays.}, such that the hop count from each source to the sink was at most 5. The algorithm used to form this Steiner tree is a variation of the SPTiRP algorithm (\cite{fullpaper}), and is described in the Appendix. 

\subsection{Model validation}
\subsubsection{Accuracy of the simplified vector fixed point equations}
\label{subsubsec:vector-fp-accuracy}
To verify the accuracy of the simplified vector fixed point equations~\eqref{eqn:simple-tau-i} in the uniqueness regime (described by Eqn.~\eqref{eqn:uniqueness-condition}) w.r.t the original, more involved fixed point equations described in Section~\ref{subsec:full-analysis}, we analyzed each of the 5 test networks using both the original equations, and the simplified equations for several different arrival rates within the uniqueness regime for that topology\footnote{Since we are considering equal arrival rates at all the sources, the uniqueness regime for a particular network is given by an upper bound on the arrival rate $\lambda$ determined by the R.H.S of Eqn.~\eqref{eqn:uniqueness-condition}, and the total hop count of the concerned network.} For each topology, we computed the worst case percentage error (over all the nodes, and all the arrival rates) in the simplified fixed point $\{\overline{\tau}_{-i}\}_{i=1}^N$ w.r.t their values obtained from the original analysis. We also computed the worst case error in $\max_{1\leq i\leq N}\delta_i$, the maximum packet discard probability over all the nodes in the network. Note that when $\max_{1\leq i\leq N}\delta_i\leq 0.002$, even an absolute error of 0.0005 would result in a percentage error of 25\%. We, therefore, adopt the following convention for reporting the errors in $\max_{1\leq i\leq N}\delta_i$. For arrival rates at which $\max_{1\leq i\leq N}\delta_i\leq 0.002$, we report the worst case \emph{absolute} error (over all such arrival rates, and all the nodes in the network); for arrival rates at which $\max_{1\leq i\leq N}\delta_i> 0.002$, we report the worst case \emph{percentage} error (over all such arrival rates, and all the nodes in the network). Finally, we also computed the worst case percentage error in the end-to-end probability of delivery (over all arrival rates, and all the sources in the network). Table~\ref{tbl:accuracy-vector-fp} summarizes the results.

\begin{table*}[ht]
  \centering
\caption{Worst case (over all nodes and all arrival rates) errors of the simplified vector fixed point scheme w.r.t the original fixed point scheme in the uniqueness regime}
\label{tbl:accuracy-vector-fp}
\footnotesize
  \begin{tabular}{|c|c|c|c|c|c|c|c|}\hline
    Topology & Node & Total hop & \% error  & Absolute error & \% error & \% error & \% error \\
& count & count & in $\overline{\tau}_{-i}$ &  in $\max_{1\leq i\leq N}\delta_i$ & in $\max_{1\leq i\leq N}\delta_i$ & in end-to-end & in end-to-end \\
& & & & for $\displaystyle{\max_{1\leq i\leq N}}\delta_i\leq 0.002$ & for $\displaystyle{\max_{1\leq i\leq N}}\delta_i> 0.002$ & delivery probability & delay\\
 \hline
    1 & 14 & 36 & 8.09 & .00026 & 12.5 & .1095 & 5.1\\
 \hline
    2 & 10 & 25 & 10.05 & .00007 & 0.7 & .0236 & 4.98\\
 \hline	
    3 & 12 & 26 & 9.85 & .00016 & 4.74 & .0492 & 5.7\\
 \hline
    4 & 12 & 27 & 8.29 & .00022 & 5.76 & .0739 & 4.83\\
 \hline
    5 & 12 & 22 & 9.77 & .00015 & 5.33 & .0487 & 5.6\\  
  \hline
\end{tabular}
\normalsize
\end{table*}

\gap
\noindent
\textbf{Observations:} 
\begin{enumerate}
\item The error in $\overline{\tau}_{-i}$ never exceeded 10.05\%.
\item The error in $\max_{1\leq i\leq N}\delta_i$ never exceeded 6\% in the regime where $\max_{1\leq i\leq N}\delta_i> 0.002$.
\item The error in end-to-end delay never exceeded 6\%. 
\item The error in end-to-end delivery probability was negligibly small over the entire uniqueness regime, never exceeding 0.11\%, i.e., \emph{the simplified fixed point equations predicted the end-to-end delivery probability extremely accurately}. 
\end{enumerate} 

\gap
\noindent
\remark In this paper, we have chosen to compare the results obtained from our simplified fixed point equations against those obtained from the detailed fixed point equations in \cite{srivastava} as opposed to comparing against simulation results from the original system. The reason for this choice is as follows: the detailed fixed point equations were shown to be very accurate (well within 10\% compared to simulations) in the regime where the discard probability on a link is close to the link PER. Since we are interested in this low discard regime, we can compare our results against these detailed equations instead of the more time consuming simulations. Note that in this low discard regime, an error of 10\% w.r.t the detailed equations translates to an error of at most 19\%-21\% w.r.t the original system.   

\subsubsection{Validity of the equality assumption}
\label{subsubsec:equality-validation}
The numerical experiments leading to Observation~\ref{cor:equality} were conducted as follows: we analyzed each of the 5 test networks using the full analysis (described in Section~\ref{subsec:full-analysis}) to obtain the vector $\{\overline{\tau}_{-i}\}_{i=1}^N$ for several different arrival rates ranging from 0.001 pkts/sec to 4 pkts/sec. Then, we computed Jain's fairness index \cite{fairness-ix} for each of these vectors. The Jain's fairness index is a measure of fairness (equality) among the components in a given vector $\mathbf{x}=\{x_1,\ldots, x_N\}$, and is computed as 

\begin{align}
J(\mathbf{x}) = \frac{(\sum_{i=1}^N x_i)^2}{N\sum_{i=1}^N x_i^2}
\end{align}
$J(\x)=1$ when all the components are exactly equal, and is $k/N$ when $k$ components have equal non-zero values, and the rest of the components are zero. In particular, \emph{closer the value of $J(\x)$ to 1, better is the fairness among the components}. 

As it turned out, in all our experiments, $J(\{\overline{\tau}_{-i}\}_{i=1}^N\})\gg \frac{N-1}{N}$, and within $1.5\%$ of 1, indicating that the components were indeed roughly equal. The results are summarized in Table~\ref{tbl:fairness}, where, for compact representation, we have reported, for each topology, the \emph{least} value of the fairness index over all the arrival rates.

\begin{table*}[ht]
  \centering
\caption{Validation of Observation~\ref{cor:equality}, using Jain's Fairness Index; fairness index of 1 implies exact equality}
\label{tbl:fairness}
\footnotesize
  \begin{tabular}{|c|c|c|c|}\hline
    Topology & Node count & Worst case fairness index\\
   \hline
   1 & 14 & 0.99608\\
   \hline
   2 & 10 & 0.98741\\	
  \hline
   3 & 12 & 0.99334\\
  \hline
   4 & 12 & 0.99536\\
  \hline
   5 & 12 & 0.99314\\   
  \hline
\end{tabular}
\normalsize
\end{table*}

\subsubsection{Tightness of the bound \eqref{eqn:scalar-vector-fp-relation}}
\label{subsubsec:upper-bound-validation}

To check the tightness of the bound given by Eqn.~\eqref{eqn:scalar-vector-fp-relation} in the uniqueness regime (described by Eqn.~\ref{eqn:uniqueness-condition}), we analyzed each of the 5 test networks using both the simplified vector fixed point equations~\eqref{eqn:simple-tau-i}, \eqref{eqn:simple-alpha} (to obtain $\max_{1\leq i\leq N}\overline{\tau}_{-i}$), and the scalar fixed point equations~\eqref{eqn:scalar-tau}, \eqref{eqn:scalar-alpha} (to obtain $\overline{\tau}$) for several different arrival rates within the uniqueness regime for that topology. The uniqueness regime for a particular network can be computed as described in Section~\ref{subsubsec:vector-fp-accuracy}. For each topology, we computed the worst case percentage error between $\max_{1\leq i\leq N}\overline{\tau}_{-i}$ and $\overline{\tau}$ over all the arrival rates. The results are summarized in Table~\ref{tbl:tightness-upper-bound}. We can see from the table that the bound is tight to within $3\%$ even in the worst case over all the arrival rates tested. 

\begin{table*}[ht]
  \centering
\caption{Slackness in the bound \eqref{eqn:scalar-vector-fp-relation}}
\label{tbl:tightness-upper-bound}
\footnotesize
  \begin{tabular}{|c|c|c|c|}\hline
    Topology & Node count & Worst case \% slack\\
             &            & in bound~\eqref{eqn:scalar-vector-fp-relation}\\
             &            & over all tested arrival rates\\
   \hline
   1 & 14 & 0.945\\
   \hline
   2 & 10 & 2.2\\	
  \hline
   3 & 12 & 2.055\\
  \hline
   4 & 12 & 1.912\\
  \hline
   5 & 12 & 2.801\\   
  \hline
\end{tabular}
\normalsize
\end{table*}
 
\subsection{Throughput optimality of the shortest path tree}
\label{subsec:spt-throughput}
To verify the throughput performance of shortest path trees, we generated 5 random instances, each with 10 sources, and 30 potential relay locations deployed uniformly over a $150\times 150\:m^2$ area.

\subsubsection{Comparison against competing topologies}
\label{subsubsec:competing-topologies-compare}

We are looking for a design that uses a small number of nodes and has a large throughput region for the given target QoS. Since an exhaustive search for the optimal throughput over all possible Steiner trees is computationally impractical, we proceeded as follows to compute an estimate of the optimal throughput for each instance.

Intuitively, two topological properties can affect the throughput of a given network, namely, the total hop count (i.e., the total load in the network), and the number of nodes in the network. Since the SPT gives the least total hop count, clearly there is no need to look at designs that use more relays than an SPT. Hence, for each instance, we first constructed an arbitrary shortest path tree, $T_{\mathrm{SPT}}$, connecting the sources to the sink. Let $n_{\mathrm{SPT}}$ be the number of relays in this SPT. We then generated many other candidate Steiner trees as follows: for each $n\in\mathbb{N}$ such that $0\leq n\leq n_{\mathrm{SPT}}$, we considered all possible combinations of $n$ relays, and with each of these combinations, we constructed a shortest path tree connecting the sources to the sink. If the resulting tree violated the hop constraint $h_{\max}=5$, it was discarded; otherwise, it was accepted as a candidate solution. 

The intuition behind this method of selecting candidate trees is as follows: for a deployment consisting of a \emph{fixed} number of nodes, among all possible topologies that one can construct over those nodes, the SPT achieves the least total hop count. Hence, it is likely to achieve better throughput than any other topology over those nodes. Hence, we consider only the shortest path tree over each combination of relays as a candidate solution. 

For the chosen QoS targets (see Section~\ref{subsec:experiment-setting}), we analyzed each of these candidate trees using the full analysis (Section~\ref{subsec:full-analysis}) for increasing arrival rates, starting from 0.001 pkts/sec, and obtained the maximum arrival rate up to which the target discard probability requirement was met. The maximum value of this arrival rate among all the candidate trees was taken as the estimate for the optimal throughput for that instance. Let us denote this as $\hat{\lambda}^\star$. This was compared against the throughput achieved by the initially constructed SPT, $T_{\mathrm{SPT}}$, obtained using the full analysis in the same manner as described above, and denoted by $\hat{\lambda}^{\mathrm{SPT}}$. \emph{For all the 5 instances, it turned out that $\hat{\lambda}^{\mathrm{SPT}}=\hat{\lambda}^\star$.} The results are summarized in Table~\ref{tbl:throughput-optimality}.

\begin{table*}[ht]
  \centering
\caption{Verification of throughput optimality of shortest path tree}
\label{tbl:throughput-optimality}
\footnotesize
  \begin{tabular}{|c|c|c|}\hline
    Scenario & Maximum throughput & Maximum throughput\\
             & among candidate trees & of $T_{\mathrm{SPT}}$\\
             & $\hat{\lambda}^\star$(pkts/sec) & $\hat{\lambda}^{\mathrm{SPT}}$(pkts/sec) \\
              
 \hline
    1 & 5  & 5\\
 \hline
    2 & 4 & 4\\
 \hline
    3 & 3 & 3\\
 \hline
    4 & 5 & 5\\
 \hline
    5 & 5 & 5\\
\hline
\end{tabular}
\normalsize
\end{table*}

\subsubsection{Comparison against the outcome of the SPTiRP algorithm\cite{bhattacharya-kumar14comnet}}
\label{subsubsec:comparison-against-sptirp}

For each instance, we also computed a hop count feasible Steiner tree, $T_{\mathrm{SPTiRP}}$ using the SPTiRP algorithm proposed in \cite{bhattacharya-kumar14comnet} ($T_{\mathrm{SPT}}$ was used as the initial feasible solution in running the SPTiRP algorithm). \emph{The idea is to check how much we gain in terms of throughput by using a shortest path design (without relay count constraint) instead of the SPTiRP design which uses a nearly minimum number of relays}. Hence, for this Steiner tree also, we computed its throughput using the full analysis in the same manner as described above. Furthermore, for the chosen QoS targets, we computed the inner bound on the throughput (maximum arrival rate) as predicted by Eqn.~\ref{eqn:pdel-membership} for both $T_{\mathrm{SPT}}$ and $T_{\mathrm{SPTiRP}}$.\footnote{Note that for the other candidate trees, we did not compute the inner bound as it is clear from Eqn.~\ref{eqn:spt-throughput-optimality} that the SPT will have the maximum value of this inner bound among all the candidate solutions.} The results are summarized in Table~\ref{tbl:spt-throughput}. 

\begin{table*}[ht]
  \centering
\caption{Throughput comparison of the shortest path tree and the SPTiRP design}
\label{tbl:spt-throughput}
\footnotesize
  \begin{tabular}{|c|c|c|c|c|}\hline
    Scenario & \multicolumn{2}{c|}{Predicted $\lambda_{\max}$(pkts/sec)} & \multicolumn{2}{c|}{$\lambda_{\max}$(pkts/sec)} \\
             & \multicolumn{2}{c|}{from formula} & \multicolumn{2}{c|}{from full analysis}\\
              
             &  $T_{\mathrm{SPTiRP}}$ & $T_{\mathrm{SPT}}$ & $T_{\mathrm{SPTiRP}}$ & $T_{\mathrm{SPT}}$\\
 \hline
    1 & 2.605 & 3.511 & 3.5 & 5\\
 \hline
    2 & 2.375 & 2.785 & 3 & 4\\
 \hline
    3 & 2.019 & 2.243 & 2.9 & 3\\
 \hline
    4 & 3.106 & 3.67 & 4 & 5\\
 \hline
    5 & 2.884 & 3.67 & 4 & 5\\
\hline
\end{tabular}
\normalsize
\end{table*}

From Table~\ref{tbl:spt-throughput}, we observe the following:

\begin{enumerate}
\item As we had predicted, the shortest path tree always achieved better throughput than the outcome of the SPTiRP algorithm.
\item However, even the SPTiRP design \emph{can operate at a significant positive load}, while possibly saving on the number of relays used. Hence, when relays are costly, the design based on the SPTiRP algorithm can be used instead of the shortest path tree without much loss in throughput.
\item Finally, comparing column~2 against column~4, and column~3 against column~5, we observe that the inner bound on the throughput region ($\tLambda_{\deltabar,\dbar}$) predicted by our formula are within 30\% of the throughput region obtained using the detailed analysis ($\hLambda_{\deltabar,\dbar}$).
\end{enumerate}

\subsubsection{Effect of slight fluctuation of arrival rates about the throughput}
In the previous experiments, we have observed that when the arrival rates at all the sources are equal, the shortest path tree is approximately throughput optimal, i.e., it can sustain a higher arrival rate than all other competing topologies without violating QoS. In this section, we intend to see what happens to the QoS performance of the shortest path tree when the arrival rates at the sources are independently subjected to small deviations from the maximum equal arrival rate that the topology can carry without violating QoS (i.e., its throughput).   

To do this, we proceeded as follows. For each instance, we first obtained the throughput of the shortest path tree in the manner described in Section~\ref{subsubsec:competing-topologies-compare}. See, for example, column~3 of Table~\ref{tbl:throughput-optimality}. For half of the sources, their arrival rates were incremented from the equal throughput by a value chosen independently and uniformly from the set $\{0.01, 0.02, 0.03, 0.04, 0.05\}$. For the remaining half of the sources, their arrival rates were decremented from the equal throughput by a value chosen independently and uniformly from the set $\{0.01, 0.02, 0.03, 0.04, 0.05\}$. Consider, for example, Scenario~1 from Table~\ref{tbl:throughput-optimality}. The maximum equal throughput is 5 packets/sec. For 5 sources, the arrival rates were chosen uniformly and independently from the set $\{5.01, 5.02, 5.03, 5.04, 5.05\}$ packets/sec. For the remaining 5 sources, the arrival rates were chosen uniformly and independently from the set $\{4.99, 4.98, 4.97, 4.96, 4.95\}$ packets/sec. For the resulting arrival rate vector, the shortest path tree topology was analyzed using the detailed fixed point analysis, and the resulting maximum discard probability and maximum mean delay (over all links) were observed, and compared against the target values, namely, $\deltabar = 0.0208$ and $\dbar = 20$ msec. For brevity, we summarize the results for 5 instances (the same ones as in Table~\ref{tbl:throughput-optimality}) in Table~\ref{tbl:spt-test-unequal-arrival-rates}, where we also report for each instance, the average arrival rate (averaged over the sources), as well as the corresponding `equal' throughput (obtained in Table~\ref{tbl:throughput-optimality}). Our observations for all the 60 scenarios tested are summarized at the end of this section.

\begin{table*}[ht]
  \centering
\caption{Worst case (over all links) values of discard probability and single-hop mean delay for the shortest path tree under slight deviations in arrival rates from the maximum sustainable equal arrival rate}
\label{tbl:spt-test-unequal-arrival-rates}
\scriptsize
  \begin{tabular}{|c|c|c|c|c|c|c|}\hline
    Scenario & Maximum sustainable & Mean & Worst case & \% error & Worst case & \% error\\
          & equal arrival rate       & arrival rate & discard probability & w.r.t target & mean single-hop & w.r.t target\\
         &   (pkts/sec)                 &  (pkts/sec)  &                               &                    &   delay (in msecs) &                  \\
 \hline
    1 & 5 & 4.997 & 0.0198 & 0 & 9.8 & 0\\
 \hline
    2 & 4 & 4.006 & 0.0210 & 0.3745 & 9.7 & 0\\
 \hline
    3 & 3 & 3.007 & 0.0147 & 0 & 9.2 & 0 \\
 \hline
    4 & 5 & 5.003 & 0.0158 & 0 & 9.3 & 0\\
 \hline
    5 & 5 & 5.004 & 0.0159 & 0 & 9.3 & 0\\
\hline
\end{tabular}
\normalsize
\end{table*} 

\textbf{Observations}
\begin{enumerate}
\item The mean delay requirement was never violated in any of the 60 scenarios tested.
\item Except in Scenario 2, the discard probability target was never violated in any of the 60 scenarios tested.
\item Even in Scenario 2, the violation in the discard probability requirement was only by an insignificant value of 0.3745\%.
\end{enumerate}
Thus, we can conclude that small deviations in the arrival rates from the maximum sustainable equal arrival rate have no significant impact on the QoS performance of the shortest path tree.

\subsubsection{Verification by practical experiments}

So far, we have verified our analytical prediction through numerical experiments and simulations. It is, however, of interest to see whether the predictions hold good in an actual deployment. First recall that we had argued in Section~\ref{sec:equal-arrival-rates} that for the chosen QoS targets, if the packet delivery probability requirement is met for every source at an arrival rate, then the delay requirement is also met. A similar conclusion was also drawn from our simulation results in \cite{srivastava}. Thus, it suffices to verify if the packet delivery probability requirement is met at the maximum arrival rate predicted by our analysis. To this end, we proceeded as follows. We created 5 arbitrary tree topologies, each with no hidden nodes, in a small area inside our lab. The details of each of the topologies are given in Table~\ref{tbl:field-trials}. We used TelosB motes running TinyOS-2.x operating system for the experiments.  We used a packet size of 131 bytes. Owing to the close proximity of the motes in our deployments, the link packet error rate turns out to be very small, namely, 0.00001 at the worst. For each topology, we first obtained its throughput using the detailed analysis (\cite{srivastava}) in the manner described in Section~\ref{subsubsec:competing-topologies-compare}. Then, for each topology, we generated packets from each source at that arrival rate according to a Poisson process, and routed the packets to the sink along the routes specified by the tree. The experiment was terminated when at least 3000 packets were generated from each source. To avoid computational and memory overheads at the nodes, instead of measuring the discard probability on each link, we measured the packet delivery ratio for each source at the base station, and compared that to the predicted packet delivery ratio for the target per link packet discard probability of 0.0208. Table~\ref{tbl:field-trials} summarizes the results, where, for brevity, we provide for each topology, the worst case packet delivery ratio (\emph{over all sources}). 

\begin{table*}[ht]
  \centering
\caption{Results from experiments with actual motes; column 7 reports the measured worst case packet delivery probability \emph{over all sources} for each topology}
\label{tbl:field-trials}
\scriptsize
  \begin{tabular}{|c|c|c|c|c|c|c|c|}\hline
    Scenario &Number of & Total number & Total hop & Maximum hop & Predicted throughput  & Measured & Predicted\\
          & sources ($m$) & of nodes  & count ($\sum_{k=1}^m h_k$) & count ($\max_{1\leq k\leq m}h_k$) & from analysis (pkts/sec) & worst case & delivery\\
          &               &           &            &             &                       &  delivery probability & probability\\
         
 \hline
    1 & 4 & 8 & 10 & 3 & 8.1 & 0.962 & 0.9376\\
 \hline
    2 & 5 & 13 & 20 & 4 & 6 & 0.936 & 0.9168\\
 \hline
    3 & 8 & 14 & 18 & 3 & 6 & 0.945 & 0.9376\\
 \hline
    4 & 4 & 11 & 24 & 7 & 5 & 0.928 & 0.8544\\
 \hline
    5 & 7 & 14 & 20 & 3 & 6 & 0.947 & 0.9376\\
\hline
\end{tabular}
\normalsize
\vspace{-4.5mm}
\end{table*} 
We observe from Table~\ref{tbl:field-trials} that for each of the tested topologies, the measured packet delivery probability at the arrival rate predicted by our analysis actually exceeds the packet delivery probability predicted from the target per link discard probability and the maximum number of hops, thus further validating the analytical model. 

\subsection{Sensitivity of the bounds to protocol parameters}
Finally, it is interesting to ask how the bounds derived in Sections~\ref{sec:arguments-leading-to-thm2} and \ref{sec:arguments-leading-to-thm3} vary with the parameters of the protocol. To check this, we proceeded as follows: first observe that the protocol parameters that affect the bounds are (i) $n_c$, the maximum number of CCA failures before packet discard, (ii) $n_t$, the maximum number of transmission failures before packet discard, (iii) the minimum back-off exponent \cite{IEEE}, which determines the range from which the first random back-off, and hence the subsequent random back-offs are sampled; \emph{this affects only the mean delay, and hence the bound $B^{\prime}(\deltabar,\dbar)$, but does not affect $B(\deltabar)$}.

To test the effect of each of these protocol parameters, we fixed the other protocol parameters at their default values (as specified in \cite{IEEE}), and varied the concerned parameter over a reasonable range, and for each value of the parameter, we computed $B_1(\deltabar)$, $B_2(\deltabar)$, $B(\deltabar)$, and $B^{\prime}(\deltabar,\dbar)$. The results are summarized in Tables~\ref{tbl:sensitivity-to-nc}, \ref{tbl:sensitivity-to-nt}, and \ref{tbl:sensitivity-to-minbo}. 

\begin{table*}[ht]
  \centering
\caption{Variation of the bounds with $n_c$, keeping the other protocol parameters fixed at their default values}
\label{tbl:sensitivity-to-nc}
\footnotesize
  \begin{tabular}{|c|c|c|c|c|}\hline
    $n_c$ & $B_1(\deltabar)$ & $B_2(\deltabar)$ & $B(\deltabar)$ & $B^{\prime}(\deltabar,\dbar)$\\
 \hline
    3 & 67.11 & 66 & 66 & 122.21\\
 \hline
    4 & 92.64 & 91 &	91	& 101.28\\
 \hline
    5 & 80.75 & 110.5 & 80.75 & 82.85\\
\hline
    6 & 62.22 & 126 & 62.22 &	67.18\\
\hline
\end{tabular}
\normalsize
\end{table*}  

\begin{table*}[ht]
  \centering
\caption{Variation of the bounds with $n_t$, keeping the other protocol parameters fixed at their default values}
\label{tbl:sensitivity-to-nt}
\footnotesize
  \begin{tabular}{|c|c|c|c|c|}\hline
    $n_t$ & $B_1(\deltabar)$ & $B_2(\deltabar)$ & $B(\deltabar)$ & $B^{\prime}(\deltabar,\dbar)$\\
 \hline
    2 & 80.75 & 107 & 80.75 & 85.32\\
 \hline
    3 & 80.75 & 110.5 &	80.75 & 82.85\\
 \hline
    4 & 80.75 & 110.5 & 80.75 & 82.85\\
\hline
    5 & 80.75 & 110.5 & 80.75 & 82.7\\
\hline
\end{tabular}
\normalsize
\end{table*}  

\begin{table*}[ht]
  \centering
\caption{Variation of the $B^{\prime}(\deltabar,\dbar)$ with minimum backoff exponent, keeping the other protocol parameters fixed at their default values}
\label{tbl:sensitivity-to-minbo}
\footnotesize
  \begin{tabular}{|c|c|c|c|c|}\hline
    Minimum backoff exponent & $B^{\prime}(\deltabar,\dbar)$\\
 \hline
    1 & 164.95\\
 \hline
    2 & 128.93\\
 \hline
    3 & 82.85\\
\hline
    4 & 35.67\\
\hline
\end{tabular}
\normalsize
\end{table*} 

\gap
\noindent
\textbf{Discussion:}

\begin{enumerate}
\item From Table~\ref{tbl:sensitivity-to-nc}, we observe the following:
\begin{enumerate}
\item $B_1(\deltabar)$ first decreases, then increases in $n_c$. This can be explained as follows: from the expression for $B_1(\deltabar)$ in \eqref{eqn:uniqueness-condition}, it is easy to show that the first term inside the $\min$ is monotonically increasing in $n_c$; it can also be shown using any mathematical toolbox (we used MATLAB R2011b) that the second term inside the $\min$ is monotonically decreasing in $n_c$ for $n_c \geq 1$. Moreover, these two functions cross over when $n_c\in(4,5)$, and for $n_c\leq 4$, the first term governs $B_1(\deltabar)$. This explains the observation that $B_1(\deltabar)$ increases with $n_c$ up to $n_c = 4$, and decreases thereafter.  
\item $B_2(\deltabar)$ is monotonically increasing in $n_c$. Intuitively, this is expected since increasing $n_c$ lowers the probability of packet discard due to successive CCA failures. Hence, the total attempt rate in the network (and therefore, the total load $\sum_{k=1}^m \lambda_k h_k$) can be pushed to a higher value before the discard probability target is violated.
\item $B(\deltabar)=\min\{B_1(\deltabar),B_2(\deltabar)\}$ increases in $n_c$ up to $n_c = 4$, then decreases in $n_c$. Note also from Table~\ref{tbl:sensitivity-to-nc} that up to $n_c = 4$, $B(\deltabar)$ equals $B_2(\deltabar)$ which is increasing in $n_c$, but thereafter, $B(\deltabar)$ equals $B_1(\deltabar)$ which is decreasing in $n_c$ for $n_c > 4$.
\item $B^{\prime}(\deltabar,\dbar)$ is monotonically decreasing in $n_c$. Intuitively, this is expected since increasing $n_c$ allows a packet to occupy the HOL position for longer duration, thereby increasing queueing delay, and causing the mean delay target to be violated sooner. 
\end{enumerate}
\item From Table~\ref{tbl:sensitivity-to-nt}, we observe the following:
\begin{enumerate}
\item $B_1(\deltabar)$ is \emph{invariant} of $n_t$. The reason is clear from Eqn.~\ref{eqn:uniqueness-condition}, where we see that the expression for $B_1(\deltabar)$ does not involve $n_t$.
\item $B_2(\deltabar)$ first increases in $n_t$, but then flattens off. Intuitively, as $n_t$ increases, the probability of packet discard due to transmission failure decreases, and hence the total attempt rate in the network can be pushed higher; thus, $B_2(\deltabar)$ should increase in $n_t$. However, when there are no hidden nodes, probability of a collision (which is only due to simulataneous channel sensing) is small; in addition, if the packet error rate on a link is small, the packet discard on a link is essentially governed by CCA failures, rather than transmission failures. Hence, increase in $n_t$ does not have a pronounced effect on $B_2(\deltabar)$.
\item For the entire tested range of $n_t$, $B(\deltabar)$ is governed by $B_1(\deltabar)$, and hence is invariant of $n_t$.
\item $B^{\prime}(\deltabar,\dbar)$ is monotonically decreasing in $n_t$. This is intuitive, since increasing $n_t$ causes a packet to occupy the HOL position longer, thereby increasing queueing delay, and violating the mean delay target sooner. The effect, however, is not as pronounced as with $n_c$, due to reasons explained in the previous point regarding $B_2(\deltabar)$. 
\end{enumerate}
\item From Table~\ref{tbl:sensitivity-to-minbo}, we observe that $B^{\prime}(\deltabar,\dbar)$ is monotonically decreasing in the minimum backoff exponent. Intuitively, this is expected since increasing the minimum backoff exponent causes the mean backoff duration of a packet to increase, thus increasing the mean delay experienced by a packet, and causing the mean delay target to be violated early. 
\end{enumerate}

\section{Conclusion}

In this paper, we have studied the problem of QoS aware network design under a class of CSMA/CA protocols including IEEE~802.15.4 CSMA/CA. Assuming that there are no hidden terminals in the network, we have derived a simplified set of fixed point equations from the more elaborate analysis developed in \cite{srivastava}. We have proved the uniqueness of the fixed point of our proposed equations, and verified through numerical experiments, their accuracy w.r.t the detailed analysis in \cite{srivastava}. From these simplified equations, we have derived an approximate inner bound on the QoS respecting throughput region of a given tree network where the QoS requirements are, a target discard probability on each link, and a target mean delay on each link (obtained by equally splitting the end-to-end QoS objectives over the links in a source-sink path). The structure of our inner bound sheds light on the dependence of the network performance on topological properties, and the arrival rate vector. In particular, our results indicate that to achieve a target per hop discard probability (respectively, mean delay), it suffices to control the total load (respectively, the maximum load over all nodes) in the network. Furthermore, for the special case of default backoff parameters of IEEE~802.15.4, equal arrival rates at all sources, and reasonable values of QoS targets, we have argued that controlling the total load is enough to meet both the discard probability and the delay targets. Under the same special case, we have also shown that the shortest path tree achieves the maximum throughput among all topologies that satisfy the approximate sufficient condition to meet the QoS targets. 

In our ongoing work, we aim to extend these results to the more general case where there are hidden terminals in the network.  

\section{Appendix}
\subsection{An algorithm for Steiner tree construction with hop constraint and relay constraint}

\begin{enumerate}
\item Run the SPTiRP algorithm (proposed in \cite{iwqos}) on the graph $G$ until the relay count in the resulting solution is $\leq N_{\max}$, or the hop constraint is violated for some source. If the resulting solution is feasible, retain it as a candidate solution, and compute its total hop count.
\item Repeat Step 1 for a fixed number of iterations, each time starting with a randomly chosen SPT as the initial solution. 
\item Pick the candidate solution with the least total hop count. 
\end{enumerate}

\subsection{Proof of Proposition~\ref{prop:vector-uniqueness}}
\label{subsec:proof-vector-uniqueness}
We begin by establishing the following result, which will be used to prove the main result.

\begin{lemma}
\label{lem:uniqueness-lipschitz}
Let $f=(f_1, f_2,\ldots,f_N):[0,a]^N\rightarrow [0,a]^N$ be such that for all $\mathbf{x}\in [0,a]^N$, and for all $i = 1,\ldots,N$, $f_i(\mathbf{x})=\sum_{j\neq i}g_j(x_j)$, where $g_j:[0,a]\rightarrow \mathbb{R}_+$ is Lipschitz continuous with Lipschitz constant $L_j$ for each $j$, $1\leq j\leq N$, and $\sum_{j=1}^N L_j < 1$. Then, the fixed point equations $\mathbf{x}=f(\mathbf{x})$ have a unique solution in $[0,a]^N$.
\end{lemma}
\begin{proof}
By our hypothesis, $f$ maps $[0,a]^N$ into $[0,a]^N$. We shall show that under the hypotheses, $f$ is a contraction on $[0,a]^N$ w.r.t the metric $d$ which is defined as $d(\mathbf{x},\mathbf{y})=\displaystyle{\max_{1\leq i\leq N}}|x_i - y_i|\:\forall \mathbf{x}, \mathbf{y}\in [0,a]^N$. Then, since $[0,a]^N$ is complete w.r.t the metric $d$, the uniqueness of the fixed point follows from Banach's fixed point theorem \cite{rudin}.

Let $\x_1,\x_2\in [0,a]^N$. Then, for all $i=1,\ldots,N$, 
\begin{align}
|f_i(\x_1)-f_i(\x_2)|&=\Bigg|\sum_{j\neq i}(g_j(x_{1,j})-g_j(x_{2,j}))\Bigg|\nonumber\\
&\leq \sum_{j\neq i}|g_j(x_{1,j})-g_j(x_{2,j})|\nonumber\\
&\leq \sum_{j\neq i}L_j|x_{1,j}-x_{2,j}|\label{eqn:use-of-lipschitz}\\
&\leq \Bigg(\sum_{j\neq i}L_j\Bigg)\max_{1\leq j\leq N}|x_{1,j}-x_{2,j}|\nonumber\\
&= \Bigg(\sum_{j\neq i}L_j\Bigg)d(\x_1,\x_2)\nonumber\\
&\leq \Bigg(\sum_{j=1}^N L_j\Bigg)d(\x_1,\x_2)\nonumber
\end{align}
where, in \eqref{eqn:use-of-lipschitz}, we have used the Lipschitz continuity of $g_j(\cdot)$. 

It follows that 
\begin{align}
d(f(\x_1),f(\x_2))&= \max_{1\leq i\leq N}|f_i(\x_1)-f_i(\x_2)|\nonumber\\
&\leq (\sum_{j=1}^N L_j)d(\x_1,\x_2)\nonumber\\
&\eqqcolon Kd(\x_1,\x_2)\nonumber
\end{align}
where, $K=\sum_{j=1}^N L_j\:<\:1$ by the hypothesis of the lemma. Thus, $f$ is a contraction on $[0,a]^N$. This completes the proof of the lemma. 
\end{proof}

\gap
\noindent
\textit{Proof of Proposition~\ref{prop:vector-uniqueness}:}

Using the notation of Lemma~\ref{lem:uniqueness-lipschitz}, let, for all $i=1,\ldots,N$, $j=1,\ldots,N$, $x_i=\overline{\tau}_{-i}$, $g_j(x_j)=\nu_j(1+\alpha_j+\ldots+\alpha_j^{n_c-1})$, with $\alpha_j=\frac{T_{tx}x_j}{1+T_{tx}x_j}$, and $f_i(\mathbf{x})=\sum_{j\neq i}g_j(x_j)$. Then, our simplified fixed point equations \eqref{eqn:simple-tau-i}, \eqref{eqn:simple-alpha} are of the form $\mathbf{x}=f(\mathbf{x})$, where $f=(f_1,\ldots,f_N)$. 

Let $\mathbf{x}\in[0,a]^N$ with $a=\frac{\alpha_{\max}}{T_{\mathsf{tx}}(1-\alpha_{\max})}$ and $\alpha_{\max}=\deltabar^{1/n_c}$. Then, $\alpha_j\leq \alpha_{\max}$ for all $j=1,\ldots,N$, and we have
\begin{align}
f_i(\mathbf{x})&= \sum_{j\neq i}\nu_j(1+\alpha_j+\ldots+\alpha_j^{n_c-1})\nonumber\\
&\leq (1+\alpha_{\max}+\ldots\alpha_{\max}^{n_c-1})\sum_{j\neq i}\nu_j\nonumber\\
&\leq (1+\alpha_{\max}+\ldots\alpha_{\max}^{n_c-1})\sum_{j=1}^N \nu_j\nonumber\\
&= (1+\alpha_{\max}+\ldots\alpha_{\max}^{n_c-1})\sum_{k=1}^m \lambda_kh_k,\:\text{using (A1)}\nonumber\\
&\leq a\nonumber
\end{align}
where, in the last step, we have used the hypothesis that $\sum_{k=1}^m \lambda_kh_k < \frac{a}{1+\alpha_{\max}+\ldots\alpha_{\max}^{n_c-1}}$. Thus, $f(\cdot)$ maps $[0,a]^N$ \emph{into} $[0,a]^N$. 

Next we investigate the functions $g_j(\cdot)$, $j=1,\ldots,N$. Let $\mathbf{x}_1,\mathbf{x}_2\in[0,a]^N$. Let $\alpha_{1,j}$ and $\alpha_{2,j}$ be the values of $\alpha_j$ corresponding to $\mathbf{x}_1$ and $\mathbf{x}_2$ respectively. Then, $\alpha_{1,j},\alpha_{2,j}\leq \alpha_{\max}$ for all $j=1,\ldots,N$, and we have
\begin{align}
|g_j(x_{1,j})-g_j(x_{2,j})|&=\nu_j|\alpha_{1,j}-\alpha_{2,j}|\times(1+\sum_{k=1}^{n_c-2}\sum_{l=0}^k\alpha_{1,j}^{k-l}\alpha_{2,j}^l)\nonumber\\
&\leq \nu_j(1+2\alpha_{\max}+\ldots +(n_c-1)\alpha_{\max}^{n_c-2})\nonumber\\
&\times|\alpha_{1,j}-\alpha_{2,j}|\nonumber\\
&\leq \nu_jT_{tx}(1+2\alpha_{\max}+\ldots +(n_c-1)\alpha_{\max}^{n_c-2})\nonumber\\
&\times|x_{1,j}-x_{2,j}|\nonumber\\
&= L_j|x_{1,j}-x_{2,j}|\nonumber
\end{align}
where, $L_j\eqqcolon \nu_jT_{tx}(1+2\alpha_{\max}+\ldots +(n_c-1)\alpha_{\max}^{n_c-2})$. Thus, $g_j(\cdot)$ is Lipschitz continuous with Lipschitz constant $L_j$ for all $j=1,\ldots, N$. Moreover, we have
\begin{align}
\sum_{j=1}^N L_j &= T_{tx}(1+2\alpha_{\max}+\ldots +(n_c-1)\alpha_{\max}^{n_c-2})\sum_{j=1}^N \nu_j\nonumber\\
&= T_{tx}(1+2\alpha_{\max}+\ldots +(n_c-1)\alpha_{\max}^{n_c-2})\sum_{k=1}^m \lambda_kh_k\nonumber\\
&< 1\nonumber
\end{align} 
where the last step follows from the condition that $\sum_{k=1}^m \lambda_kh_k < \frac{1}{T_{tx}(1+2\alpha_{\max}+\ldots +(n_c-1)\alpha_{\max}^{n_c-2})}$.

Thus, we see that all the conditions of Lemma~\ref{lem:uniqueness-lipschitz} are satisfied by the simplified fixed point equations in $[0,a]^N$. Hence, by Lemma~\ref{lem:uniqueness-lipschitz}, the simplified fixed point equations have a unique solution in $[0,a]^N$.\hfill \Square

\subsection{Proof of Lemma~\ref{lem:tau}}
\label{subsec:proof-lemma-tau}
\begin{proof}
For concreteness, we prove the result using the default backoff parameters of IEEE~802.15.4 CSMA/CA. However, the proof goes through for any protocol in this class.
(i) Observe that $\overline{\tau}_j^{(i)}=\frac{\beta_j \times b_j \times q_j}{1 - q_j + q_j \times b_j}\geq \beta_j \times b_j \times q_j$. From Equation~\eqref{eqn:nh-beta}, it can be verified that $\beta_j$ is a continuous, monotonically non-increasing function of $\alpha_j\in[0,1]$, and hence, for all $j$, $\beta_j \geq \frac{1}{238}\eqqcolon \underline{\beta}>0$ (per symbol time), where the last inequality is obtained by evaluating Equation~\eqref{eqn:nh-beta} at $\alpha_j = 1$. 

From Equation~\eqref{eqn:nh-calc_of_b}, it follows that for all $j$, $b_j \geq \frac{\overline{B}_j}{\overline{B}_j + T_{\mathsf{tx}}}=\frac{1}{1+\frac{T_{tx}}{\overline{B}_j}}\geq \frac{1}{1+\frac{T_{tx}}{78}}\eqqcolon \underline{b}>0$, since $\overline{B}_j\geq 78$ symbol time. 

Finally, notice that for all $j$, the mean service time per HOL packet, $\frac{1}{\sigma_j}$ must satisfy $\frac{1}{\sigma_j}> 78$ symbol time (the mean duration for the first back-off and CCA), since the HOL packet must spend at least one back-off duration before it can be transmitted or discarded. Hence, for all $j$, $q_j = \min\{1,\frac{\nu_j}{\sigma_j}\}\geq \min\{1,\frac{\min_{1\leq k\leq m}\lambda_k}{\sigma_j}\}\geq\min\{1,78\min_{1\leq k\leq m}\lambda_k\}\eqqcolon \underline{q}(\mathbf{\lambda})>0$, where $\mathbf{\lambda}$ is expressed in packets/symbol time. 

Combining the above bounds, we have, for all $j$, and for any given positive $\{\lambda_k\}_{k=1}^m$,
\begin{equation}
\overline{\tau}_j^{(i)}\geq \beta_j \times b_j \times q_j\geq \underline{\beta}\underline{b}\underline{q}(\mathbf{\lambda})> 0
\end{equation}
as claimed.

\noindent
(ii) For all $j$, $\overline{\tau}_j^{(i)}=\frac{\beta_j \times b_j \times q_j}{1 - q_j + q_j \times b_j}\leq \frac{\beta_j \times b_j \times q_j}{q_j \times b_j}=\beta_j\leq \frac{1}{78}/\text{symbol time}$, where the last inequality follows by observing that $\beta_j$ is monotonically non-increasing in $\alpha_j\in[0,1]$, and hence evaluating Equation~\eqref{eqn:nh-beta} at $\alpha_j = 0$. 
\end{proof}

\subsection{Proof of Proposition~\ref{prop:scalar-uniqueness}}
\label{subsec:proof-scalar-uniqueness}
To prove the existence and uniqueness of the scalar fixed point equation, we shall need one more analytical result, which we state below.

\begin{lemma}
\label{lem:scalar-uniqueness-lipschitz}
Let $a\in\mathbb{R}_+$. Suppose $f:[0,a]\rightarrow [0,a]$ such that $f$ can be expressed as $f(\cdot)= Mg(\cdot)$, where $M > 0$, and $g:[0,a]\rightarrow \mathbb{R}_+$ is a Lipschitz continuous function with Lipschitz constant $L$ such that $ML < 1$. Then,
\begin{enumerate}
\item The fixed point equation $x=f(x)$ has a unique solution in $[0,a]$.
\item Let us denote this unique fixed point as $h(M)$, to indicate its dependence on $M$. Suppose now that $g(\cdot)$ is differentiable on $(0,a)$. Then, $h(M)$ is continuous, and monotonically increasing in $M$. 
\item Suppose further that $g\in \mathcal{C}^1$, i.e., $g^{\prime}(\cdot)$ is continuous. Then, $h(\cdot)$ is differentiable. 
\end{enumerate}
\end{lemma} 

\begin{proof}
\begin{enumerate}
\item By hypothesis, $f$ maps $[0,a]$ into $[0,a]$. We shall show that under the condition $ML < 1$, $f(\cdot)$ is a contraction on $[0,a]$ w.r.t $|\cdot|$. Then, since $[0,a]$ is a complete metric space, it follows from Banach's fixed point theorem that the fixed point equation has a unique solution in $[0,a]$. 

Consider $x_1,x_2\in [0,a]$. Then,
\begin{align}
\Bigg|f(x_1)-f(x_2)\Bigg|&=M|g(x_1)-g(x_2)|\nonumber\\
&\leq ML|x_1-x_2|\nonumber\\
&=C|x_1 - x_2|\nonumber
\end{align}
where, $C\eqqcolon ML < 1$. Thus, $f$ is a contraction. Hence the proof. 
\item We have, $h(M)= Mg(h(M))$. Then, for $M > 0$, and non-zero $\epsilon$ (small enough such that $M+\epsilon > 0$),
\begin{align}
h(M+\epsilon) - h(M) &= (M+\epsilon)g(h(M+\epsilon)) - Mg(h(M))\nonumber\\
&= M[g(h(M+\epsilon))-g(h(M))] \nonumber\\
&+ \epsilon g(h(M+\epsilon))\nonumber\\
&= M[\{h(M+\epsilon)-h(M)\}g^{\prime}(\xi)]\nonumber\\
& + \epsilon g(h(M+\epsilon))\label{eqn:mvt}\\
&= Mg^{\prime}(\xi)[h(M+\epsilon)-h(M)]\nonumber\\
& + \epsilon g(h(M+\epsilon))\nonumber
\end{align}
where, Equation~\ref{eqn:mvt} follows from the Mean Value Theorem, with $\xi$ lying between $h(M)$ and $h(M+\epsilon)$. Thus, we have,
\begin{align}
h(M+\epsilon) - h(M) &= \frac{\epsilon g(h(M+\epsilon))}{1-Mg^{\prime}(\xi)}\nonumber\\
&= \frac{\epsilon \frac{h(M+\epsilon)}{M+\epsilon}}{1-Mg^{\prime}(\xi)}\label{eqn:difference-eqn}
\end{align}
Note that for $\epsilon > 0$, the R.H.S of Equation~\eqref{eqn:difference-eqn} is \emph{non-negative}, since $Mg^{\prime}(\xi)\leq M|g^{\prime}(\xi)|\leq ML < 1$, and $h(\cdot)\geq 0$. Thus, it follows that $h(M)$ is monotonically increasing in $M$. 

To show that $h(\cdot)$ is continuous in $M$, we proceed as follows. Since the denominator of the R.H.S of Eqn.~\ref{eqn:difference-eqn} is non-negative for $ML < 1$, we have, from Equation~\ref{eqn:difference-eqn}, 
\begin{align}
|h(M+\epsilon) - h(M)| &= \frac{|\epsilon|\Bigg|\frac{h(M+\epsilon)}{M+\epsilon}\Bigg|}{1-Mg^{\prime}(\xi)}\nonumber\\
&\leq \frac{\Bigg|\frac{\epsilon}{M+\epsilon}\Bigg||h(M+\epsilon)|}{1-ML}\nonumber\\
&\leq \frac{\Bigg|\frac{\epsilon}{M+\epsilon}\Bigg|}{1-ML}a,\:\text{since }h(M+\epsilon)\in[0,a]\nonumber
\end{align}
Thus, for any given $\delta > 0$, $|h(M+\epsilon) - h(M)|<\delta$ whenever $\Bigg|\frac{\epsilon}{M+\epsilon}\Bigg| < \frac{1-ML}{a}\delta \eqqcolon t(\delta)$. Without loss of generality\footnote{any $\epsilon$ that works for such a $\delta$ also works for higher values of $\delta$}, we take $\delta < \frac{a}{1-ML}$ so that $t(\delta)<1$. Then, $\Bigg|\frac{\epsilon}{M+\epsilon}\Bigg| < t(\delta)\Leftrightarrow -t(\delta)<\frac{\epsilon}{M+\epsilon}<t(\delta)$, which yields, after some algebraic manipulations, $\frac{-Mt(\delta)}{1+t(\delta)} < \epsilon < \frac{Mt(\delta)}{1-t(\delta)}$. This holds whenever $|\epsilon|< \frac{Mt(\delta)}{1+t(\delta)}$. Hence, the continuity of $h(\cdot)$ follows.

\item From Eqn.\ref{eqn:difference-eqn}, we have
\begin{align}
\frac{h(M+\epsilon) - h(M)}{\epsilon} &= \frac{\frac{h(M+\epsilon)}{M+\epsilon}}{1-Mg^{\prime}(\xi)}\nonumber
\end{align} 
Taking limits on both sides as $\epsilon \rightarrow 0$, and using the continuity of $h(\cdot)$ and $g^{\prime}(\cdot)$, we have
\begin{align}
h^{\prime}(M)&=\frac{h(M)/M}{1 - Mg^{\prime}(h(M))}
\end{align}
since $\xi$ lies between $h(M+\epsilon)$ and $h(M)$. This completes the proof. 
\end{enumerate}
\end{proof}

\gap
\noindent
\textit{Proof of Proposition~\ref{prop:scalar-uniqueness}:}
Using the notation of Lemma~\ref{lem:scalar-uniqueness-lipschitz}, let us write $x=\overline{\tau}$, $M=\sum_{k=1}^m \lambda_kh_k$, $g(x)=1+\alpha+\ldots +\alpha^{n_c-1}$, and $f(x)=Mg(x)$, where $\alpha = \frac{T_{tx}x}{1+T_{tx}x}$. Then, using a derivation similar to the one in the proof of Proposition~\ref{prop:vector-uniqueness}, it can be seen that $f(\cdot)$ maps $[0,a]$ into $[0,a]$ under the condition $\sum_{k=1}^m\lambda_kh_k < \frac{a}{1+\alpha_{\max}+\ldots+\alpha_{\max}^{n_c-1}}$.

Let $x_1,x_2\in [0,a]$. Then,

\begin{align}
|g(x_1)-g(x_2)| &= |\alpha_1-\alpha_2|(1+\sum_{k=1}^{n_c-2}\sum_{l=0}^k\alpha_{1}^{k-l}\alpha_{2}^l)\nonumber\\
&\leq (1+2\alpha_{\max}+\ldots + (n_c-1)\alpha_{\max}^{n_c-2})|\alpha_1-\alpha_2|\nonumber\\
&\leq T_{tx}(1+2\alpha_{\max}+\ldots + (n_c-1)\alpha_{\max}^{n_c-2})|x_1 - x_2|\nonumber\\
&= L|x_1 - x_2|\nonumber
\end{align}  
where, $L\eqqcolon  T_{tx}(1+2\alpha_{\max}+\ldots + (n_c-1)\alpha_{\max}^{n_c-2})$. Thus, $g(\cdot)$ is Lipschitz continuous with Lipschitz constant $L$. Moreover, under the condition given by Eqn.~\ref{eqn:uniqueness-condition}, $ML = (\sum_{k=1}^m\lambda_kh_k) T_{tx}(1+2\alpha_{\max}+\ldots + (n_c-1)\alpha_{\max}^{n_c-2}) < 1$. Hence, from Part~1 of lemma~\ref{lem:scalar-uniqueness-lipschitz}, it follows that the scalar fixed point equation defined by \eqref{eqn:scalar-tau}, \eqref{eqn:scalar-alpha} has a unique solution in $[0,a]$.

\subsection{Proof of Lemma~\ref{lem:delta-inc-alpha-gamma}}
\label{subsec:proof-delta}
\begin{proof}
Clearly, for a fixed value of $\alpha_i$, $\delta_i$ is non-decreasing in $\gamma_i$. Now, let us \emph{fix} $\gamma_i$, and study the behavior of $\delta_i$ as a function of $\alpha_i$ \emph{alone}. In the following discussion, we omit the subscript $i$ whenever it is immediate from the context.

Let $x=\alpha^{n_c}$. Then we have, from Equation~\eqref{eqn:nh-delta},
\begin{align}
\delta &= [x + \gamma^{n_t}(1-x)^{n_t}] + \gamma x(1-x) + \ldots+ \gamma^{n_t-1} x(1-x)^{n_t-1}\nonumber
\end{align}
For any fixed $\gamma$, it is easy to see (by taking derivatives with respect to $x$) that the last $(n_t-1)$ terms are non-decreasing in $x$ if $x\leq 1/n_t$. Also, for $\gamma^{n_t}\leq 1/n_t$, the derivative of the first term w.r.t $x$ is $1 - n_t\gamma^{n_t}(1-x)^{n_t-1}\geq 1-n_t\gamma^{n_t}\geq 0$. Thus, for fixed $\gamma$, $\delta$ is non-decreasing in $x=\alpha^{n_c}$, and hence in $\alpha$ if $\alpha^{n_c} \leq 1/n_t$, and $\gamma^{n_t}\leq 1/n_t$. Hence, for $\alpha_{i,1}\leq \alpha_{i,2}$, and $\gamma_{i,1}\leq \gamma_{i,2}$, we have
\begin{align}
\delta_{i,1}&=\delta(\alpha_{i,1},\gamma_{i,1})\nonumber\\
&\leq \delta(\alpha_{i,1},\gamma_{i,2})\nonumber\\
&\leq \delta(\alpha_{i,2},\gamma_{i,2})\nonumber\\
&= \delta_{i,2}\nonumber
\end{align}
\end{proof}

\subsection{Proof of Proposition~\ref{prop:scalar-vector-fp-relation}}
\label{subsec:proof-scalar-vector}

\begin{proof}
Recall from the proof of Proposition~\ref{prop:vector-uniqueness} that the function describing the vector fixed point equations is a contraction on $[0,a]^N$, where $a$ is as defined in Sections~\ref{subsec:vector-uniqueness} and \ref{subsec:scalar-uniqueness}. Hence, an iterative procedure starting with any initial solution in $[0,a]^N$ will converge to the unique fixed point $\{\overline{\tau}_{-i}\}_{i=1}^N$. 

Let us initialize the iteration with $\overline{\tau}_{-i}^{(0)}=\overline{\tau}$ for all $i=1,\ldots,N$. Note that this initial solution is in $[0,a]^N$. Then, from eqns.~\eqref{eqn:simple-alpha} and \eqref{eqn:scalar-alpha}, we have $\alpha_i^{(0)}=\alpha$ for all $i=1,\ldots,N$. Thus, for all $i=1,\ldots,N$,
\begin{align}
\overline{\tau}_{-i}^{(1)}&= \sum_{j\ne i}\nu_j(1+\alpha_j^{(0)}+\ldots+(\alpha_j^{(0)})^{n_c-1}),\:\text{where }\nu_j=\sum_{k=1}^mz_{k,j}\lambda_k \nonumber\\
&=(1+\alpha+\ldots+\alpha^{n_c-1})\sum_{j\ne i}\nu_j\nonumber\\
&\leq (1+\alpha+\ldots+\alpha^{n_c-1})\sum_{j=1}^N\nu_j\nonumber\\
&= (\sum_{k=1}^m\lambda_kh_k)(1+\alpha+\ldots+\alpha^{n_c-1})\nonumber\\
&= \overline{\tau}\nonumber\\
&= \overline{\tau}_{-i}^{(0)}\label{eqn:induction-base-step}
\end{align}

Now, we use induction to complete the proof. Suppose, for some iteration $k$, we have, for all $i=1,\ldots,N$, $\overline{\tau}_{-i}^{(k)}\leq \overline{\tau}_{-i}^{(k-1)}$. Then, since the R.H.S of Eqn.~\ref{eqn:simple-alpha} is monotonically increasing in $\overline{\tau}_{-i}$, it follows that for all $i=1,\ldots,N$, $\alpha_i^{(k)}\leq \alpha_i^{(k-1)}$. Hence, for all $i=1,\ldots, N$,
\begin{align}
\overline{\tau}_{-i}^{(k+1)}&= \sum_{j\ne i}\nu_j(1+\alpha_j^{(k)}+\ldots+(\alpha_j^{(k)})^{n_c-1})\nonumber\\
&\leq \sum_{j\ne i}\nu_j(1+\alpha_j^{(k-1)}+\ldots+(\alpha_j^{(k-1)})^{n_c-1})\nonumber\\
&= \overline{\tau}_{-i}^{(k)}\nonumber
\end{align}
Thus, when the monotonicity property of the iterates holds for iteration $k$, it also holds for iteration $(k+1)$. From Eqn.~\ref{eqn:induction-base-step}, it holds for $k=1$. Hence, by induction, it holds for all iterations. Since the iterates converge to the unique fixed point $\{\overline{\tau}_{-i}\}_{i=1}^N$, it follows that for all $i=1,\ldots,N$, $\overline{\tau}_{-i}\leq \overline{\tau}_{-i}^{(0)}=\overline{\tau}$. This completes the proof.
\end{proof}

\footnotesize
\bibliography{dit_astec}

\begin{thebibliography}{10}

\bibitem{bhattacharya-kumar14comnet}
A.~Bhattacharya and A.~Kumar, ``A shortest path tree based algorithm for relay
  placement in a wireless sensor network and its performance analysis,'' {\em
  Computer Networks}, vol.~71, pp.~48--62, 2014.

\bibitem{bhattacharya-kumar14tr-qos-aware-nw-design-csma}
A.~Bhattacharya and A.~Kumar, ``An approximate inner bound to the {QoS} aware
  throughput region of a tree network under {IEEE}~802.15.4 {CSMA/CA} and
  application to wireless sensor network design,'' tech. rep., available at
  arxiv.org/pdf/1408.1222, 2014.

\bibitem{srivastava}
R.~Srivastava, S.~M. Ladwa, A.~Bhattacharya, and A.~Kumar, ``A fast and
  accurate performance analysis of beaconless {IEEE}~802.15.4 multi-hop
  networks,'' tech. rep., available at arxiv.org/pdf/1408.1225, 2014.

\bibitem{IEEE}
IEEE, {\em IEEE Standards Part 15.4: Wireless Medium Access Control (MAC) and
  Physical Layer (PHY) Specifications for Low-Rate Wireless Personal Area
  Networks (LR-WPANs)}.
\newblock New York, October 2003.

\bibitem{klein}
P.~N. Klein and R.~Ravi, ``A nearly best-possible approximation algorithm for
  node-weighted steiner trees,'' {\em Journal of Algorithms}, vol.~19, no.~1,
  pp.~104--115, 1995.

\bibitem{Lin}
G.-H. Lin and G.~Xue, ``Steiner tree problem with minimum number of {S}teiner
  points and bounded edge length,'' {\em Information {P}rocessing {L}etters},
  vol.~69, pp.~53--57, 1999.

\bibitem{ZhangHou}
H.~Zhang and J.~Hou, ``Maintaining {S}ensing {C}overage and {C}onnectivity in
  {L}arge {S}ensor {N}etworks,'' {\em Ad Hoc and Sensor Wireless Networks},
  vol.~1, pp.~89--124, March 2005.

\bibitem{Misra}
S.~Misra, S.~D. Hong, G.~Xue, and J.~Tang, ``Constrained {R}elay {N}ode
  {P}lacement in {W}ireless {S}ensor {N}etworks to {M}eet {C}onnectivity and
  {S}urvivability {R}equirements,'' in {\em I{EEE INFOCOM}}, 2008.

\bibitem{Costa}
A.~M. Costa, J.-F. Cordeau, and G.~Laporte, ``Fast heuristics for the steiner
  tree problem with revenues, budget and hop constraints,'' {\em European
  Journal of Operational Research}, vol.~190, pp.~68--78, 2008.

\bibitem{fullpaper}
A.~Bhattacharya and A.~Kumar, ``Qo{S} {A}ware and {S}urvivable {N}etwork
  {D}esign for {P}lanned {W}ireless {S}ensor {N}etworks,'' tech. rep.,
  available at arxiv.org/pdf/1110.4746, 2013.

\bibitem{smartconnect-paper}
A.~Bhattacharya, S.~M. Ladwa, R.~Srivastava, A.~Mallya, A.~Rao, D.~G.~R. Sahib,
  S.~Anand, and A.~Kumar, ``Smartconnect: A system for the design and
  deployment of wireless sensor networks,'' in {\em 5th International
  Conference on Communication Systems and Networks (COMSNETS)}, 2013, available
  at \url{http://ieeexplore.ieee.org/xpl/articleDetails.jsp?arnumber=6465582}.

\bibitem{multisink}
A.~Bhattacharya, A.~Rao, K.~P. Naveen, P.~P. Nishanth, S.~Anand, and A.~Kumar,
  ``{QoS} constrained optimal sink and relay placement in planned wireless
  sensor networks,'' in {\em 10th IEEE International Conference on Signal
  Processing and Communications (SPCOM), Bangalore, India}, 2014.

\bibitem{rachitpaper}
R.~Srivastava and A.~Kumar, ``Performance {A}nalysis of {B}eacon-{L}ess
  {IEEE}~802.15.4 {M}ulti-{H}op {N}etworks,'' in {\em 4th International
  Conference on Communication Systems and Networks (COMSNETS)}, 2012.

\bibitem{bordenave}
C.~Bordenave, D.~McDonald, and A.~Proutiere, ``Asymptotic stability region of
  slotted aloha,'' {\em IEEE TRANSACTIONS ON INFORMATION THEORY}, vol.~58,
  no.~9, pp.~5841--5855, 2012.

\bibitem{bianchi00performance}
G.~Bianchi, ``Performance analysis of the {IEEE} 802.11 distributed
  coordination function,'' {\em IEEE Journal on Selected Areas in
  Communications}, vol.~18, pp.~535--547, 2000.

\bibitem{kumar-etal04new-insights}
A.~Kumar, E.~Altman, D.~Miorandi, and M.~Goyal, ``New insights from a fixed
  point analysis of single cell {IEEE}~802.11 wireless {LAN}s,'' {\em IEEE/ACM
  Transactions on Networking}, vol.~15, pp.~588--601, 2007.

\bibitem{jindal09}
A.~Jindal and K.~Psounis, ``The achievable rate region of 802.11-scheduled
  multihop networks,'' {\em IEEE/ACM TRANSACTIONS ON NETWORKING}, vol.~17,
  no.~4, 2009.

\bibitem{park13}
P.~Park, C.~Fischione, and K.~H. Johansson, ``{M}odeling and {S}tability
  {A}nalysis of {H}ybrid {M}ultiple {A}ccess in the {IEEE}~802.15.4
  {P}rotocol,'' {\em ACM Transactions on Sensor Networks}, vol.~9, no.~2, 2013.

\bibitem{marbach11}
P.~Marbach, A.~Eryilmaz, and A.~Ozdaglar, ``Asynchronous csma policies in
  multihop wireless networks with primary interference constraints,'' {\em IEEE
  Transactions on Information Theory}, vol.~57, no.~6, 2011.

\bibitem{jain-etal03interference}
K.~Jain, J.~Padhye, V.~Padmanabhan, and L.~Qiu, ``Impact of interference on
  multi-hop wireless network performance,'' in {\em ACM MobiCom}, 2003.

\bibitem{lin-shroff04rate-control}
X.~Lin and N.~B. Shroff, ``Joint rate control and scheduling in multihop
  wireless networks,'' in {\em 43rd IEEE Conference on Decision and Control},
  2004.

\bibitem{lin-shroff06imperfect-scheduling}
X.~Lin and N.~B. Shroff, ``The impact of imperfect scheduling on cross-layer
  congestion control in wireless networks,'' {\em IEEE/ACM Transactions on
  Networking}, vol.~14, no.~2, pp.~302--315, 2006.

\bibitem{jiang-walrand12delay-distributed-scheduling}
L.~Jiang and J.~Walrand, ``Stability and delay of distributed scheduling
  algorithms for networks of conflicting queues,'' {\em Queueing Systems},
  vol.~72, pp.~161--187, 2012.

\bibitem{kmk1}
A.~Kumar, D.~Manjunath, and J.~Kuri, {\em Communication Networking}.
\newblock Morgan Kaufmann Publishers, 2004.

\bibitem{abhijitmeth}
A.~Bhattacharya, ``{N}ode {P}lacement and {T}opology {D}esign for {P}lanned
  {W}ireless {S}ensor {N}etworks,'' {M}.{E} thesis, Indian Institute of
  Science, June 2010.

\bibitem{neely10stochastic-nw-optimization}
M.~J. Neely, {\em Stochastic Network Optimization with Application to
  Communication and Queueing Systems}.
\newblock Morgan \& Claypool, 2010.

\bibitem{telemetry1}
Cyan, ``http://www.cyantechnology.com/apps/index.php.''

\bibitem{telemetry2}
B.~Aghaei, ``Using wireless sensor network in water, electricity and gas
  industry,'' in {\em 3rd IEEE International Conference on Electronics Computer
  Technology}, pp.~14--17, April 2011.

\bibitem{whitt83queueing-network-analyzer}
W.~Whitt, ``The queueing network analyzer,'' {\em The Bell System Technical
  Journal}, vol.~62, 1983.

\bibitem{ross07prob-models}
S.~M. Ross, {\em Introduction to Probability Models}.
\newblock Academic Press, 9~ed., 2007.

\bibitem{fairness-ix}
R.~Jain, D.~M. Chiu, and W.~Hawe, ``A quantitative measure of fairness and
  discrimination for resource allocation in shared computer systems,'' Tech.
  Rep. TR-301, DEC Research Report, 1984.

\bibitem{iwqos}
A.~Bhattacharya and A.~Kumar, ``Delay constrained optimal relay placement in
  planned wireless sensor networks,'' in {\em 18th IEEE International Workshop
  on Quality of Service (IWQoS), Beijing, China}, 2010.

\bibitem{rudin}
W.~Rudin, {\em Principles of Mathematical Analysis}.
\newblock McGraw-Hill Book Company, 3~ed., 1976.

\end{thebibliography}
\end{document}